\documentclass[hidelinks,12pt]{article}

\usepackage[utf8]{inputenc}
\usepackage[T1]{fontenc}
\usepackage[english]{babel}
\usepackage{csquotes}
\usepackage{authblk} 
\usepackage{ulem}  
\usepackage{lineno}

\usepackage{graphicx}
\usepackage{subcaption}
\usepackage{multirow}
\usepackage{array}
\usepackage{float}
\usepackage{geometry}
\usepackage{diagbox}
\usepackage{enumitem}
\usepackage{fontsize}
\usepackage{hyperref}
\usepackage{rotating}
\usepackage{xcolor} 
\usepackage{ulem} 
\usepackage{lmodern}

\usepackage{natbib}

\usepackage{amsmath, amsthm, amssymb, amsfonts}
\usepackage{mathrsfs, yhmath}

\newcommand{\R}{\mathbb{R}}

\newtheorem{thm}{Theorem}[section] 
\newtheorem{cor}[thm]{Corollary}
\newtheorem{prop}[thm]{Proposition}

\newtheorem{hypothesis}{Hypothesis}

\theoremstyle{definition}
\newtheorem{defi}[thm]{Definition}
\newtheorem{ex}[thm]{Example}

\theoremstyle{remark} 
\newtheorem{remark}[thm]{Remark}

\DeclareMathOperator{\cun}{\mathcal{C}^1}

\title{\textsc{Dynamics of Two Species with Density-Dependent Interactions and Application to Mutualism\thanks{This is a preprint version of an article published in \textit{Bulletin of Mathematical Biology}. The final authenticated version is available online at: \url{https://doi.org/10.1007/s11538-026-01724-1}.}}}

\author[1]{Chlo\"e Mian*}
\author[2]{Sylvain Billiard}
\author[3]{Violaine Llaurens}
\author[1,4]{Charline Smadi}

\affil[1]{Univ. Grenoble Alpes, CNRS, Institut Fourier (UMR 5582), 38000 Grenoble, France}
\affil[2]{Univ. Lille, CNRS, UMR 8198 – Evo-Eco-Paleo, F-59000 Lille, France}
\affil[3]{Collège de France, CNRS, INSERM, Centre Interdisciplinaire de Recherche en Biologie (UMR 7241), 75005 Paris, France}
\affil[4]{Univ. Grenoble Alpes, INRAE, LESSEM, 38000 Grenoble, France}

\date{}

\begin{document}

\maketitle

\vspace{-3em}
\begin{center}
*Corresponding author. E-mail: chloe.mian@univ-grenoble-alpes.fr
\end{center}


\renewcommand{\abstractname}{\textsc{Abstract}}

\begin{abstract}
Ecological interactions shape the dynamics of natural populations in the wild. 
Density-dependent processes are widespread and may change the respective effects of populations on one another, for instance by shifting interactions from mutualistic to parasitic relationships. 
Here, we develop a general deterministic model of two interacting populations, assuming density-dependent costs and benefits for the interacting individuals within and between species. 
This framework aims at generalizing pre-existing population dynamics models involving competition, predation, mutualism and parasitism, by allowing ecological interactions to transition when the respective densities of interacting species change. 
Through ordinary differential equations and phase portrait analysis, we derive general principles governing these systems, identifying constraints on the organization of equilibria and sufficient conditions for the emergence of certain dynamic behaviors.
In particular, we show that equilibrium indices alternate along isoclines under broad geometric assumptions, and that limit cycles can arise when interactions include mutualistic and parasitic phases, while they cannot be generated locally in strictly mutualistic regions where the relevant interaction signs remain fixed. 
This framework provides a general approach for characterizing the population dynamics of interacting species and highlights the effect of the density-dependent transitions in ecological interactions.
\end{abstract} 



\section{\textsc{Introduction}}
\label{sec:Intro}


In natural communities, ecological interactions between individuals from different species are widespread and profoundly shape the dynamics of populations. 
These interactions can be either antagonistic relationships, such as competition or parasitism, or cooperative ones, such as mutualism \citep{bronstein2015a, janzen1985}, therefore leading to different demographic outcomes in the populations involved. 
The costs and benefits incurred to individuals belonging to species involved in ecological interactions usually crucially depend on the respective densities of interacting species: while interaction might be mutualistic at low density, they may become competitive or parasitic when local species densities are larger. 
For example, in ant-plant mutualism, plants provide ants with food in exchange for protection. At low ant densities, plants may not receive sufficient protection, while at high densities, the cost of supporting the ants may outweigh the protective benefits \citep{bronstein1994, holland2002, janzen1985}. 
Similarly, in pollination systems, the investment of plants in nectar can become inefficient if pollinator densities are too low to ensure effective pollination, or too high, leading to increased competition among pollinators and diminishing returns for the plants \citep{kawai2007, morris2010, holland2002}.
Moreover, the density-dependent costs and benefits incurred to the individuals might not be symmetrical. 
For example, mycorrhizal fungi provide plants with essential nutrients, while plants provide carbohydrates to the interacting fungi \citep{bonfante2010mechanisms}, but the respective effects of these exchanges on the reproduction and/or survival of individuals are not necessarily balanced between plants and fungi.

Although ecological interactions are widely recognized to be density-dependent, a general theoretical framework for studying the coexistence and dynamics of interacting species is still lacking. 
In this paper, we derive general mathematical results for two-species systems with density-dependent ecological interactions. 
We then focus in particular on mutualistic systems, as mutualistic interactions are increasingly recognized to be density-dependent, shifting along a continuum that ranges from competition to parasitism and predation. 
Mutualism therefore provides a relevant case study for investigating how density-dependent interactions shape population dynamics.

Many classical models addressing mutualism assumed strictly positive interactions, leading to unrealistic outcomes like unbounded population growth or infinite densities \citep{brauer2012, hale2021ecological, may1972}. 
Indeed, because of the complexity and context-dependence of mutualistic interactions, modeling them as a direct extension of Lotka–Volterra systems overlooks key mechanisms, such as density-dependent saturation, that prevent unbounded growth in actual populations.
Moreover, the factors shaping the persistence of mutualistic interactions and the transitions between parasitic and mutualistic interactions driven by populations densities variations (\textit{i.e.} on ecological time scales) are poorly documented in mutualistic interactions, while they are well-known in predator–prey systems modeled using Lotka-Volterra equations \citep{holland2002, revilla2015}.
\cite{brauer2012} and \cite{holland2002}  indeed considered mutualistic models with saturation effects preventing either unbounded growth or systematic extinction. 
More recently \cite{ruiz2025interaction} analyzed mutualistic systems with a density-dependent framework, emphasizing the role of benefit–cost ratios in shaping long-term interaction outcomes.
Asymmetries in obligate \textit{vs.} facultative mutualistic interactions among species have also been shown to significantly affect system stability and resilience in a series of models \citep{bronstein1994, holland2010, neuhauser2004mutualism}. 
Most of these previous models studying the dynamics of mutualistic interactions relied on ordinary differential equations to describe the conditions for coexistence and stability (see \cite{hale2021ecological} for a review). 
Other approaches such as network analysis have also been used to explore interspecific relationships within ecological systems \citep{bascompte2003nested} and partial differential equations to model the spatial dynamics of mutualistic interactions \citep{aliyu2021interplay}.


As emphasized by \cite{hale2021ecological}, it is striking that all mutualistic interaction models studied so far show similar outcomes and behaviors, suggesting that some general properties can be found, independently of the peculiarities of each model. 
Previous models of mutualistic interactions often focus on equilibrium points, typically identifying stable nodes or saddle points as the primary outcomes of these systems with particular modeling choices of the functional forms relating the rate of interactions, the birth and death rates, and the populations densities.
While these models provide valuable insights, they rarely consider the possibility of more complex dynamics, such as limit cycles or oscillatory behaviors, therefore preventing general predictions.

By building and analyzing a mathematical model of two-species interactions with general equations allowing to cover any kind of density-dependent processes, we aim to overcome these limitations and to clarify the range of possible dynamical behaviors, including equilibrium structures.
To this end, we adopt a framework that is general in two complementary senses.
First, the model proposed here is mathematically general: instead of prescribing specific functional responses or mechanistic assumptions, we rely on qualitative and structural constraints on the growth functions and the signs of their partial derivatives.
This enables the model to encompass a wide range of two-species interaction systems, including but not limited to most previously published mathematical models describing mutualism (see \cite{hale2021ecological} and Tables in Appendix \ref{app:mutualism_models}).
Second, the model described here is biologically general: because interactions are represented through density-dependent signs of the cross-partial derivatives, the framework is able to incorporate a broad spectrum of ecological processes such as mutualistic benefits, saturation, costs, and shifts to parasitic effects at high densities. 
Thus, it provides a unified tool for  ecological relationships that do not necessarily remain strictly mutualistic across all possible population densities.

This paper is structured as follows: Section \ref{sec:Model} first provides the general definitions of density-dependent effects acting on the dynamics of species involved in ecological interactions, which will be used throughout the paper.
We introduce a system of differential equations for two interacting populations, incorporating density-dependent interactions, without focusing on specific types of interactions within a particular ecological context and on limited range of parameters.
In Section \ref{sec:General_Behavior}, we formulate general and robust assumptions derived from the analysis of deterministic models (Theorem \ref{thm:equilibrium}).
We then present examples of mutualistic interaction models that, despite aiming to represent different aspects of mutualism or emphasizing certain features over others, fall within the framework of Theorem \ref{thm:equilibrium}.
In Section \ref{sec:Mutualistic}, we use the extended definition of mutualistic interactions and identify sufficient conditions for the emergence of oscillatory dynamics.
These conditions, often overlooked in previous studies, provide new insights into the potential for stable coexistence or extinction.
After discussing our results in Section \ref{sec:Discussion}, we present the corresponding proofs in Section \ref{sec:Proofs}.
\section{Model and Key Definitions}
\label{sec:Model}

\subsection{General Framework}
\label{subsec:General_Framework}

We model the dynamics of two interacting populations using the system:
\begin{equation}
\left\{
    \begin{array}{ll}
        \dot{x} &= x f\left(x,y\right) \\
        \dot{y} &= y g\left(x,y\right)
    \end{array}
\right.
\label{eq:init}
\end{equation}
where $x$ and $y$ represent the densities of each species, and $f$ and $g$ describe the density-dependent growth rates of each species.
Because we model population sizes, the analysis is confined to the non-negative quadrant of $\R^2$.
Throughout the paper, $R_+^2$ denotes the non-negative quadrant, while the strictly positive quadrant denotes its interior, where $x>0$ and $y>0$.
This form is commonly used to describe two-species dynamics without accounting for immigration (see \cite{thompson2006dynamics}), as in the models studied in \cite{may1972},  \cite{brauer2012}, and reviewed in \cite{hale2021ecological}.

Several mutualism models from the literature can be written in the form \eqref{eq:init}. 
Representative examples are summarized in Appendix~\ref{app:mutualism_models}, where they are grouped according to their underlying ecological mechanisms: linear-benefit models (\ref{subsec:linear}), saturating models (\ref{subsec:satur}), and non-monotonic or consumer-resource models (\ref{subsec:non-mono}). 
These examples are presented through their corresponding per-capita growth functions $f$ and $g$, and illustrate the range of density-dependent mechanisms covered by the framework.

Although these models differ in their biological mechanisms, they share a common phase-plane structure: the dynamics are organized by the zero-growth curves of the two species, and coexistence equilibria arise at their intersections.
Linear-benefit models typically assume that the presence of one species systematically increases the per-capita growth rate of the other, saturating models introduce bounded benefits, whereas non-monotonic and consumer-resource models allow interaction signs or strengths to vary with density. 
This motivates the sign-based framework introduced below.

\subsection{Key Assumptions and Definition}
\label{subsec:Key_Assumptions_Def}
The system of differential equations \eqref{eq:init} does not, by itself, specify the ecological nature of the interaction between the two species. 
Depending on the signs and density dependence of the functions $f$ and $g$, the same general form can describe mutualistic, parasitic, predator--prey, or competitive regimes. 
We therefore introduce below a sign-based terminology for the interaction regimes considered in this work.

To analyze and interpret the behavior of the general system \eqref{eq:init}, we follow an approach similar to \cite{brauer2012}, relying on the partial derivatives of the growth functions. 
These derivatives capture how the growth rate of each species responds to changes in the density of both species.
For this reason, starting from Section \ref{subsec:Intra_Competition}, we will assume that $f$ and $g$ are $\cun$ functions \textit{i.e.} continuous and continuously differentiable.
We focus on specific regions of the phase plane defined by the population densities $x$ and $y$. 
In each region, the signs of the partial derivatives of the interaction terms determine the local effect of one species on the growth rate of the other. 
As a result, an interaction may be positive in one part of the phase portrait and negative in another, depending on the local population densities.
Different regions of the phase portrait may correspond to different interaction regimes, described below:

\begin{itemize}
    \item \textbf{Mutualism and Parasitism}: $\frac{\partial f}{\partial y} > 0$ and $\frac{\partial g}{\partial x} > 0$ represent a region of strict mutualism, where each species positively influences the growth of the other. 
    We refer to the model as strictly mutualistic when the condition is satisfied for all population densities $x$ and $y$.
    Conversely, $\frac{\partial f}{\partial y} < 0$ or $\frac{\partial g}{\partial x} < 0$ indicates a region of parasitism, as the presence of one species negatively impacts the growth of the other.
    \item \textbf{Intraspecific Competition}: $\frac{\partial f}{\partial x} < 0$ and $\frac{\partial g}{\partial y} < 0$ describe negative feedback within each species, representing competition, where individuals of the same population limit each other's growth.
    \item \textbf{Positive Intraspecific Density Dependence and Allee effects}: $\frac{\partial f}{\partial x} > 0$ and $\frac{\partial g}{\partial y} > 0$ describe positive feedback within each species, meaning that the per-capita growth rate of a species increases with its own density, which could describe interactions as diverse as intraspecific cooperation, facilitation, mutualism or Allee effects. When such positive density dependence occurs at low population density, it corresponds to an Allee effect. We distinguish here between weak and strong Allee effects. A weak Allee effect refers to a positive relationship between population density and per-capita growth at low density, without necessarily making the per-capita growth rate negative. In contrast, a strong Allee effect occurs when the per-capita growth rate is negative below a critical density threshold, so that small populations tend to decline unless this threshold is exceeded. 
    \item \textbf{Dynamic Transitions}: The signs of $\frac{\partial f}{\partial y}$ and $\frac{\partial g}{\partial x}$ are not fixed and may change with species densities, reflecting a continuum of interactions ranging from mutualism to parasitism.
    The phase portrait may contain regions where the signs of $\frac{\partial f}{\partial y}$ and $\frac{\partial g}{\partial x}$ differ, reflecting changes in interaction type.
\end{itemize}

In contrast to many pre-existing models where these partial derivatives always remain within the same strict sign or even remain constant regardless of species densities \citep{hale2021ecological}, we allow the sign of these partial derivatives to change. 
Such assumption allows for transitions from mutualistic (positive value) to parasitic interactions (negative value) between species, depending on their respective densities. 
We thus give an extended definition of mutualism as follows:

\begin{defi}
\label{def:mut}
    A system of differential equations of the form \eqref{eq:init} will be said to be mutualistic if, in the phase portrait of $\R_+ \times \R_+$, there is at least a region where $\frac{\partial g}{\partial x} > 0$ and at least a region where $\frac{\partial f}{\partial y} > 0$. 
    These two regions may be disjoint.
\end{defi}

This definition is broader than strict mutualism, but it does not include purely competitive systems or classical predator--prey systems with fixed interaction signs. 
It requires each species to benefit from the presence of the other in at least one region of the phase portrait, in the sense that the density of the partner species increases its per-capita growth rate in that region. 
These beneficial regions may depend on population densities and do not need to occur simultaneously.
It should be understood as an extended, sign-based definition.
In particular, systems with disjoint positive interaction regions may be more naturally interpreted as mixed antagonistic or reciprocal parasitic interactions rather than strict mutualisms.

These definitions are illustrated in Figure \ref{fig:schema_def}, which presents two different situations, both consistent with mutualism as defined in Definition \ref{def:mut}. 
In the first case, a strictly mutualistic region appears, although the condition does not hold for all densities $x$ and $y$. 
In the second case, the system remains mutualistic overall, but no strictly mutualistic region exists; that is, both populations may benefit from the interaction, though not simultaneously.
The latter case may arise in ecological systems where species parasitize or feed on one another, but where the sign of the interaction depends on partner abundance, becoming beneficial to a species only for particular values of the other species density.

\begin{figure}[h]
    \centering
    \begin{subfigure}[b]{0.4\textwidth}
        \centering
        \includegraphics[width=\textwidth]{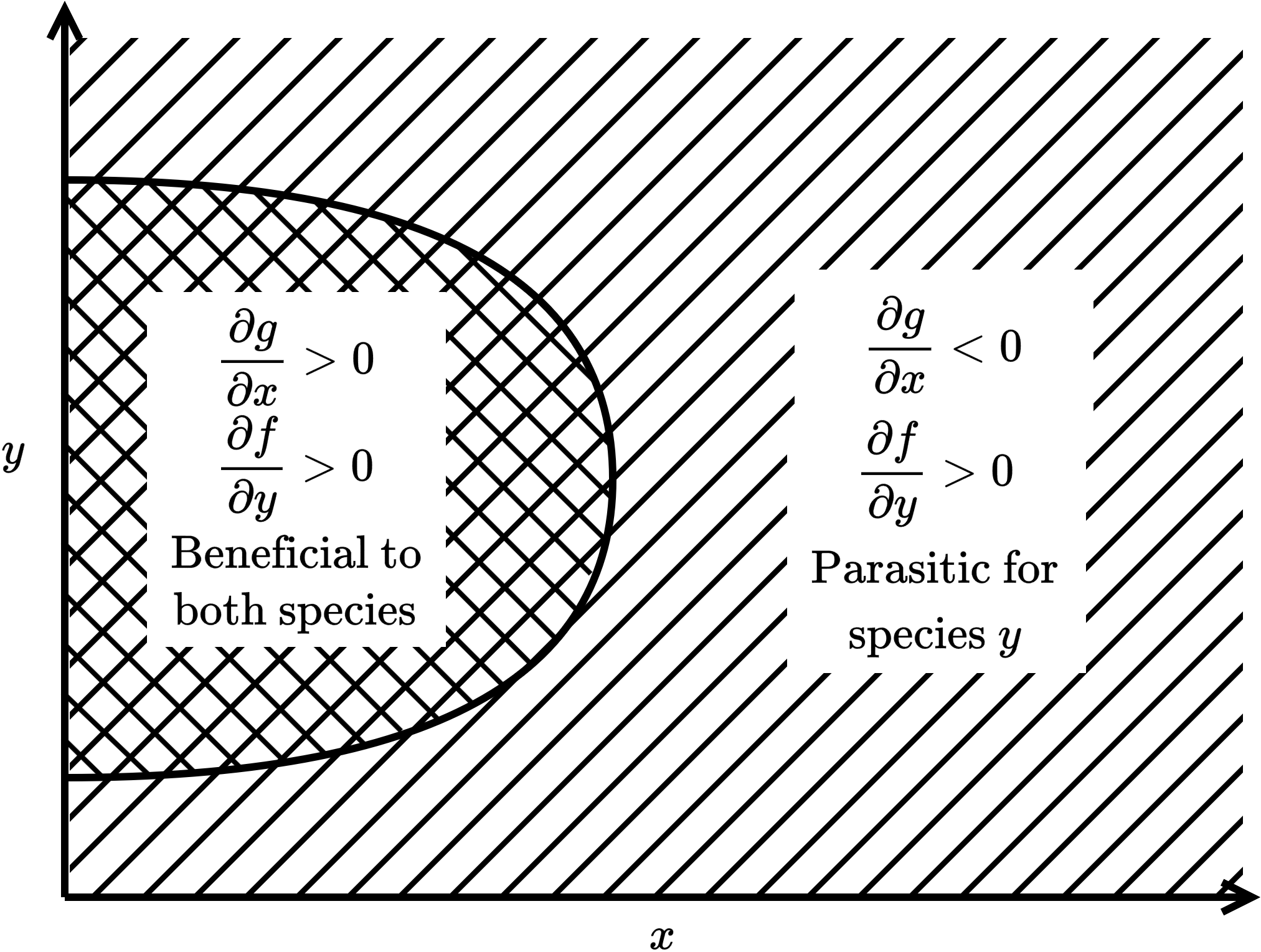}
        \caption{Local region of strict mutualism: Overlapping positive partial derivatives.}
        \label{fig:schema_mut}
    \end{subfigure}
    \hspace{1cm}
    \begin{subfigure}[b]{0.4\textwidth}
        \centering
        \includegraphics[width=\textwidth]{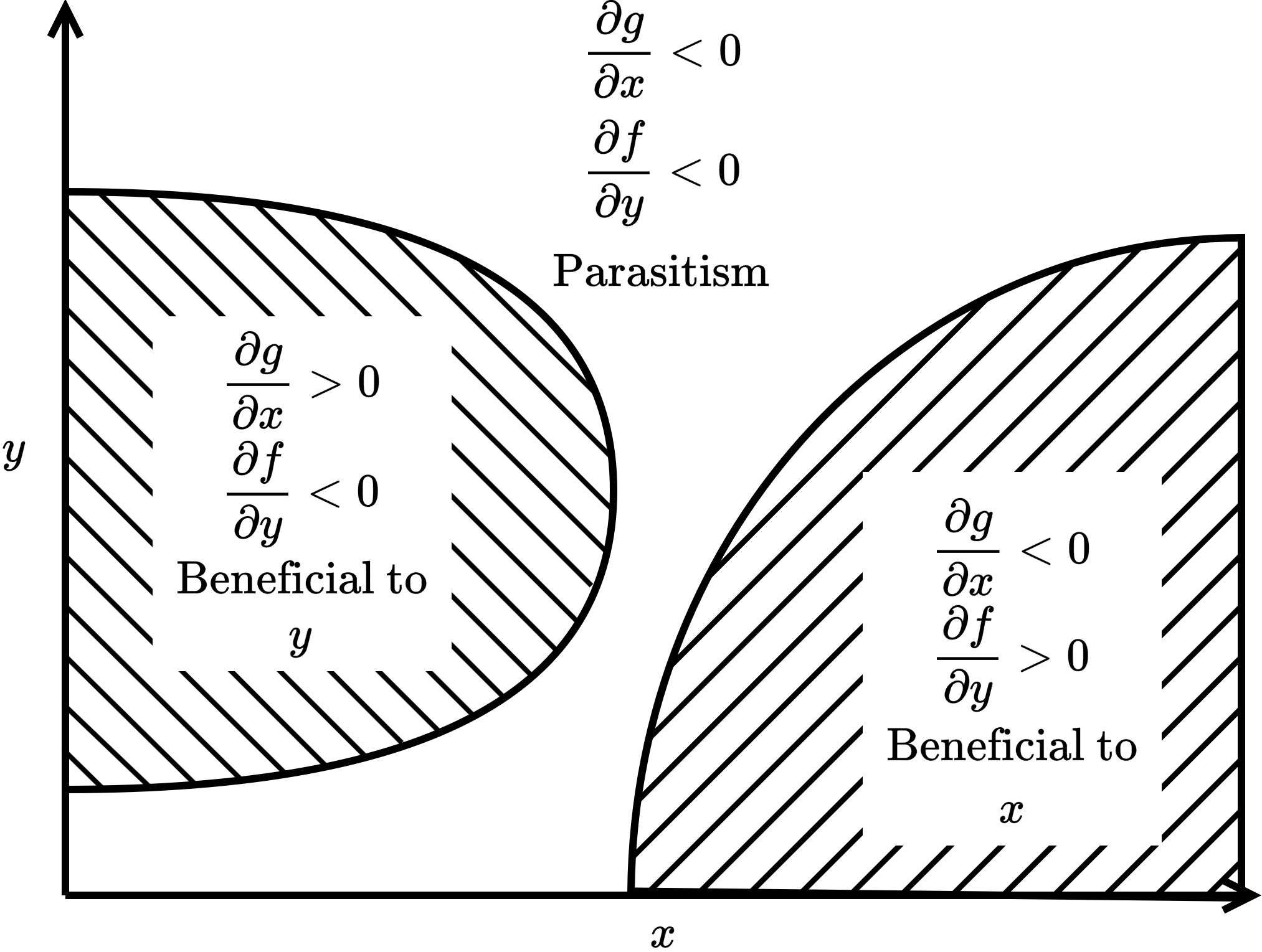}
        \caption{Disjoint regions of mutualistic influence: Non-overlapping mutualistic effects.}
        \label{fig:schema_para}
    \end{subfigure}
    \caption{\textit{Illustration of density-dependent effects of mutualism (Definition \ref{def:mut}) through the signs of cross partial derivatives.} 
    In both situations, the interaction is classified as mutualistic according to Definition \ref{def:mut}, even though the local effects of one species on the other vary.
    Figure \ref{fig:schema_mut}: $\frac{\partial f}{\partial y}>0$ everywhere, meaning that species $y$ always has a positive effect on the growth of species $x$. The cross hatched region indicates where $\frac{\partial g}{\partial x}>0$ (and then $\frac{\partial f}{\partial y}>0$), so that species $x$ also benefits species $y$ locally. This region represents a strictly mutualistic zone. Outside this zone, the effect of species $x$ on $y$ is negative, yet the system still qualifies as mutualistic under Definition \ref{def:mut}.
    Figure \ref{fig:schema_para}: The cross partial derivatives $\frac{\partial f}{\partial y}$ and $\frac{\partial g}{\partial x}$ vary and are positive in distinct, non-overlapping regions: the hatched from the top left to the bottom right (resp. hatched from the bottom left to the top right) region corresponds to $\frac{\partial g}{\partial x}>0$ (resp. $\frac{\partial f}{\partial y}>0$). Outside these regions, the interaction is parasitic for both species.}
    \label{fig:schema_def}
\end{figure}
\newpage
\section{General dynamics of two species with any type of interactions}
\label{sec:General_Behavior}

In this work, we study the behavior of system \eqref{eq:init}, including stability, oscillations, and other dynamics, with a particular focus on how the specific forms of $f$ and $g$ influence species coexistence. 
In Sections \ref{subsec:General_Assumptions}, \ref{sec:Index} and \ref{sec:Axes} we  consider a general framework in which $f$ and $g$ 
are assumed to be continuous only, and we apply index theory (see \cite{perko2013differential} and Appendix \ref{appendix:Index} for details) to prove Theorem \ref{thm:equilibrium}.
Then, in Section \ref{subsec:Intra_Competition}, using our general Theorem \ref{thm:equilibrium}, we more precisely investigate how the nature of the interactions can affect the outcomes of the system under the additional assumption that $f$ and $g$ are $\cun$. 
We mainly focus on simple equilibrium points and use methods based on the Jacobian matrix to analyze their stability and nature.
For example, regardless of the specific nature of the interaction represented by system \eqref{eq:init}, certain information can already be deduced about the equilibrium point $\left(0,0\right)$ by applying these methods (see Appendix \ref{appendix:$(0,0)$} and \cite{brauer2012}).

\subsection{Assumptions}
\label{subsec:General_Assumptions}

We introduce several assumptions about the functions $f$ and $g$ that are commonly derived from biological models of population dynamics. 
Most of these assumptions, inspired by models reviewed in \cite{hale2021ecological}, are robust and remain valid even when $f$ and $g$ are subject to minor local perturbations. 
We recall that the zeros of $f$ and $g$ define the isoclines of the system, and their intersections determine the equilibrium points.

\begin{hypothesis}\label{hyp:one}
The functions $f$ and $g$ each define a single additional isocline, beyond those at $x = 0$ and $y = 0$: one additional zero-growth curve for $f$ and one additional zero-growth curve for $g$.
The zeros $\left(x,y\right)$ of $f$ (resp. of $g$) form a single continuous curve.
\end{hypothesis}

We will denote these curves by $\Gamma_f$ and $\Gamma_g$:

\begin{equation*}
\Gamma_f := \{\left(x,y\right) \in \mathbb{R}_+ \times \mathbb{R}_+ , f\left(x,y\right)=0\}
\end{equation*}
and 

\begin{equation*}
\Gamma_g := \{\left(x,y\right) \in \mathbb{R}_+ \times \mathbb{R}_+, g\left(x,y\right)=0\},
\end{equation*}

with $\Gamma_f \neq \emptyset$ and $\Gamma_g \neq \emptyset$.

\ref{hyp:one} means that each species has a single connected zero-growth boundary in the non-negative quadrant. 
Thus, the transition between positive and negative per-capita growth is organized around one main zero-growth curve for each species, rather than around several disconnected regions.
In other words, \ref{hyp:one} excludes highly complex situations in which a species might grow at certain low densities, decline at intermediate densities, and then grow again in another region of the phase space that is disjoint from the first.

\begin{hypothesis}\label{hyp:two}
$\Gamma_f$ and $\Gamma_g$ delimit regions of strict constant sign of $f$ and $g$.
Moreover we impose that the signs of $f$ (respectively of $g$) change on either side of the curve $\Gamma_f$ (respectively of the curve $\Gamma_g$).
\end{hypothesis}

\ref{hyp:two} ensures that each isocline is not merely a zero set, but a true boundary between regions where the corresponding species increases and regions where it decreases.
Crossing the isocline $\Gamma_f$ changes the sign of the per-capita growth rate of species $x$, and crossing the isocline $\Gamma_g$ changes the sign of the per-capita growth rate of species $y$.

\begin{hypothesis}\label{hyp:three}
No more than two isoclines intersect at any given point.
\end{hypothesis}

In particular, \ref{hyp:three} implies that $\left(0,0\right) \notin \Gamma_f \cup \Gamma_g$.
\ref{hyp:three} excludes degenerate situations in which several zero-growth boundaries are reached at exactly the same population densities. 
Each equilibrium is therefore generated by the intersection of two relevant zero-growth boundaries, rather than by the simultaneous coincidence of several thresholds.

\begin{hypothesis}\label{hyp:four}
We exclude the boundary cases where the equilibrium points are formed by two isoclines that only touch at that point but do not cross, and the case where the isoclines are coincident.
\end{hypothesis}

\ref{hyp:four} excludes critical or non-robust configurations, such as tangencies where equilibria may appear or disappear under small parameter changes, or coincident isoclines that would generate a continuum of equilibria. 
Thus, the equilibria considered here correspond to robust crossings of zero-growth boundaries.\\

These assumptions directly influence the properties of the vector field.
Together, Hypotheses \ref{hyp:one} to \ref{hyp:four} define a regular phase portrait structure in which the zero-growth boundaries of the two species are well-defined, isolated, and robust under small perturbations. 
These assumptions are not intended to impose a specific ecological mechanism, but rather to identify a qualitative geometric setting shared by a broad class of two-species population models.

\subsection{Equilibrium Index Alternation}
\label{sec:Index}

One characteristic of equilibrium points is their index, which describes the local behavior of the flow in their neighborhood.
The index helps to determine whether an equilibrium point acts as an attractive, a repulsive, or a saddle point, and reveals how trajectories interact with the point (see \cite{perko2013differential} and Appendix \ref{appendix:Index} for more details).
We recall that a saddle point represents an unstable equilibrium where the system converges along one direction (stable) but diverges along another (unstable). 

From a mathematical perspective, it is a topological invariant: it is preserved under continuous perturbations of the vector field, as long as equilibria are not created, destroyed, or merged. 
An equilibrium point with index $-1$ has the local structure of a saddle point, whereas an equilibrium point with index $+1$ may correspond to an attractive point, a repulsive point, or a center. 
Thus, the index does not fully determine stability, but it strongly constrains the possible local configurations of the vector field.

From a biological perspective, saddle points of index $-1$ often act as thresholds between qualitatively different ecological outcomes. 
They may separate basins of attraction leading, for instance, to coexistence, extinction of one species, or extinction of both species. 
Equilibria of index $+1$ instead correspond to non-saddle configurations, which may represent stable coexistence when attractive, instability when repulsive, or oscillatory behavior when they are centers. 
Therefore, an alternation of indices is informative because it implies that coexistence states, instability points, and threshold states cannot be arranged arbitrarily along the isoclines.
In particular, for systems satisfying \ref{hyp:one} to \ref{hyp:four} on $f$ and $g$, the indices of equilibrium points in the positive quadrant exhibit a regular alternation along the isoclines. 
This result is formalized in the following theorem:

\begin{thm}\label{thm:equilibrium}
    Let a dynamical system in $\R_+ \times \R_+$ be described by \eqref{eq:init}, with functions $f$ and $g$ satisfying hypotheses \ref{hyp:one} to \ref{hyp:four}. 
    Then, in the strictly positive quadrant, the equilibrium points of the system alternate along the isoclines between having an index of $+1$ and an index of $-1$.
\end{thm}

We recall that an equilibrium point represents a pair of population densities that can be maintained at equilibrium. 
Theorem \ref{thm:equilibrium} implies that two non-saddle equilibria, whether attractive or repulsive, cannot be consecutive along the same isocline: a saddle equilibrium must lie between them. 
The same interpretation extends to boundary equilibria on the coordinate axes by Proposition \ref{prop:equilibrium_axes}.

Biologically, this excludes for instance a situation in which the only two possible equilibrium configurations along an isocline are a stable coexistence equilibrium, where both species persist, and a stable single-species equilibrium on an axis, where only one species persists. 
Under assumptions \ref{hyp:one} to \ref{hyp:four}, these two stable configurations must be separated by an intermediate saddle equilibrium, corresponding to another possible pair of population densities at equilibrium and acting as a threshold between the two dynamical outcomes.

The proof of this result, provided in Section \ref{sec:Proofs}, is illustrated in Figure \ref{fig:PhasePortrait_Index}.

\begin{figure}[h]
    \centering
    \begin{subfigure}[b]{0.45\textwidth}
        \centering
        \includegraphics[width=\textwidth]{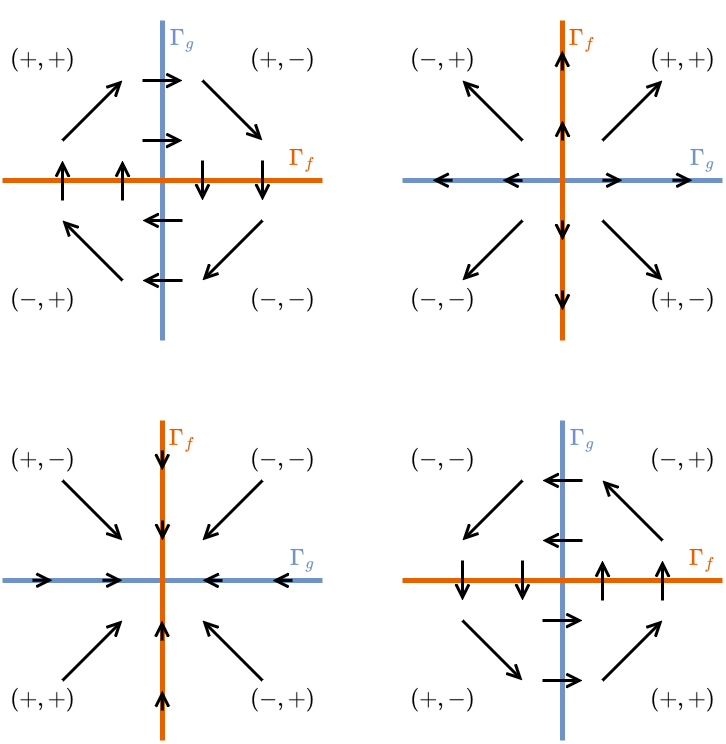}
        \caption{Local structure of the phase portrait around an index $+1$ equilibrium.}
        \label{fig:PhasePortrait_Index+1}
    \end{subfigure}
    \hfill
    \begin{subfigure}[b]{0.45\textwidth}
        \centering
        \includegraphics[width=\textwidth]{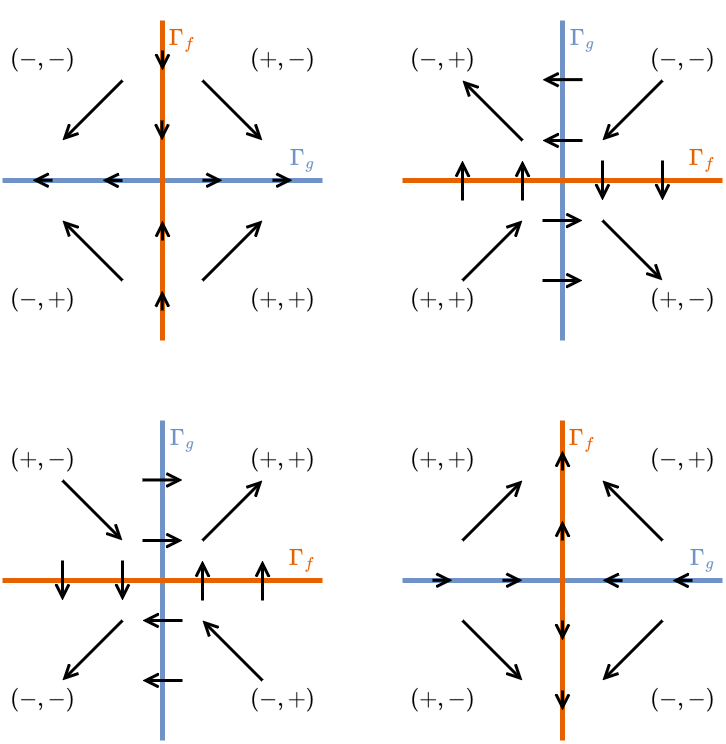}
        \caption{Local structure of the phase portrait around an index $-1$ equilibrium.}
        \label{fig:PhasePortrait_Index-1}
    \end{subfigure}
    \caption{\textit{Two possible local configurations of the phase portrait near an equilibrium point}, 
    corresponding to the two admissible clockwise sign changes of the vector field across adjacent regions.
    Illustration of Proof of Theorem \ref{thm:equilibrium}.
    The curve $\Gamma_g$ is shown in blue, and $\Gamma_f$ in orange.
    Note: A single arrow is used to represent the direction of the vector field within each region for simplicity. For instance, in the $\left(+,+\right)$ region, the direction of the vector field $\left(\nearrow\right)$ may vary from horizontal $\left(\rightarrow\right)$ (excluded) to vertical $\left(\uparrow\right)$ (excluded), as long as the signs remain strictly positive.}
    \label{fig:PhasePortrait_Index}
\end{figure}

\begin{remark}
Here, we recall that we did not assume the functions $f$ and $g$ to be $\cun$; we only required the continuity of the curves $\Gamma_f$ and $\Gamma_g$.  
For instance, we may consider a function $f$ in which the birth rate of species $x$ depends on the density of species $y$ through a threshold-like mechanism—linear at low densities of $y$, and saturating to a constant value at high densities.
\end{remark}

\begin{remark}
We recall that the model covers the wide parasitic–mutualistic continuum.
The interactions may therefore represent competition, predation, or mutualism, as long as the functions $f$ and $g$ satisfy the conditions stated in hypotheses \ref{hyp:one}–\ref{hyp:four}.
\end{remark}

Additional assumptions on $f$ and $g$ are necessary to completely describe the configuration and outcomes of the systems.
In fact, knowing that certain points have an index $+1$ is not enough to determine whether an equilibrium is attractive, repulsive, or even a center.
Moreover, identifying the nature of an equilibrium point with index $+1$ is not sufficient to determine the nature of all other equilibrium points with the same index. 
Without further assumptions on $f$ and $g$,  two equilibrium points with index $+1$ may exhibit different behaviors, with one being attractive and the other repulsive.

To illustrate these findings, we present an example of a system that can model a mutualistic interaction, and shows an alternation of equilibrium points with different characteristics.

\begin{ex}
\label{ex:ex+1}
Let us consider the following system:  

\begin{equation}
\left\{
    \begin{array}{ll}
        \dot{x} &= x f\left(x,y\right) = x\left(e_1y - d_1\left(\left(x - a_1\right)^{b_1} + c_1\right)\right), \\
        \dot{y} &= y g\left(x,y\right) = y\left(e_2x - d_2\left(\left(y - a_2\right)^{b_2} + c_2\right)\right),
    \end{array}
\right.
\label{eq:ex+1}
\end{equation}
with $(a_i, c_i, d_i, e_i, i \in \{1,2\})$ being positive constants, and $(b_i, i \in \{1,2\})$ being positive integers.

The functions $f$ and $g$ satisfy hypotheses \ref{hyp:one} to \ref{hyp:four}.  
We choose $\frac{\partial f}{\partial y} = e_1 > 0$ and $\frac{\partial g}{\partial x} = e_2 > 0$ to represent strict mutualism with a constant positive effect.  
The other parameters represent internal limitations $d_i$, constraints related to critical thresholds $a_i$, or fixed costs $c_i$.
We choose the $b_i$ to be even so that the partial derivatives
$\frac{\partial f}{\partial x}=-d_1 b_1 (x-a_1)^{b_1-1}$ and $ \frac{\partial g}{\partial y}=-d_2 b_2 (y-a_2)^{b_2-1}$ change sign at the thresholds $a_1$ and $a_2$, respectively. 
Below these thresholds, each species experiences positive intraspecific density dependence, whereas above them intraspecific limitation dominates. This corresponds to an Allee-type feedback, which is weak in general and becomes strong only under additional conditions on the sign of the per-capita growth rate at low density.

Figure \ref{fig:ex+1} displays the phase portrait corresponding to \eqref{eq:ex+1}, for certain parameter values. 
In the positive quadrant, there is an alternation of equilibrium points with index $+1$ and $-1$, with $+1$ points of different natures, including a repulsive point and an attractive point.  
This system could model mutualistic interactions where there is stable coexistence of the two species if the initial population densities are sufficiently high. 
However, if the density of one species is too low, the populations collapse, the state $(0,0)$ being attractive.

\begin{figure}[h]
    \centering
    \includegraphics[width=0.5\linewidth]{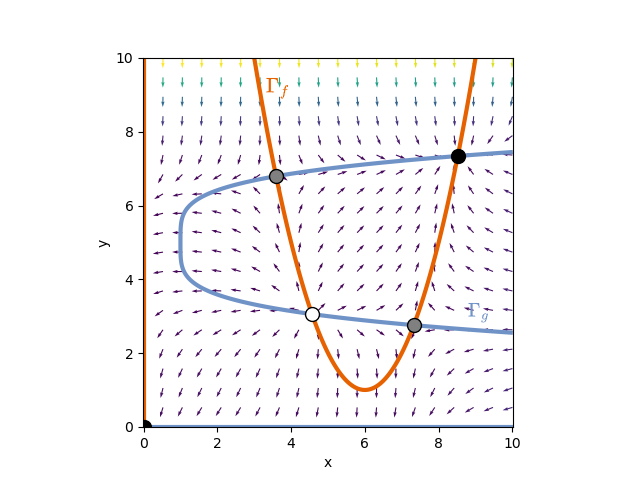}
    \caption{\textit{Example of a mutualism model exhibiting different types of equilibrium points.}
    Phase portrait of system \eqref{eq:ex+1} for the parameter values: $a_1=6$, $b_1=2$, $c_1=1$, $d_1=1$, $e_1=1$ and $a_2=5$, $b_2=4$, $c_2=1$, $d_2=1/4$, $e_2=1$.
    The curve $\Gamma_g$ is shown in blue, and $\Gamma_f$ in orange.
    White points represent repulsive points, black points represent attractive points and gray points represent saddle points.
    The figure illustrates a mutualistic system with multiple types of equilibrium points.}
    \label{fig:ex+1}
\end{figure}
\end{ex}

\subsection{Generalisation to Axes}
\label{sec:Axes}

Theorem \ref{thm:equilibrium} describes the alternation of indices for equilibrium points located strictly inside the positive quadrant, \textit{i.e.} at intersections of $\Gamma_f$ and $\Gamma_g$ with strictly positive coordinates.
However, equilibrium points may also lie on the axes. 
These additional equilibrium points may arise from the intersections of $\Gamma_f$ with the axis $x = 0$, or of $\Gamma_g$ with the axis $y = 0$.
For population models, such equilibria correspond to situations in which one species persists at equilibrium while the other is extinct.

For isoclines $\Gamma_f$ and $\Gamma_g$ satisfying hypotheses \ref{hyp:one} to \ref{hyp:four}, we restrict to configurations in which the connected zero-growth curves delimit connected sign regions in the non-negative quadrant. 
Under this restriction, each isocline intersects the coordinate axes at most twice, so that only finitely many boundary configurations have to be considered.
Taking into account the relative ordering of these intersections along the axes, a total of $71$ distinct configurations are possible.
Among them, $18$ configurations generate at least one equilibrium point located on an axis and are therefore not covered by Theorem \ref{thm:equilibrium}.
Under the same geometric assumptions on the isoclines, this alternation property can be extended to include such boundary equilibria.

\begin{prop}\label{prop:equilibrium_axes}
Under the assumptions of Theorem \ref{thm:equilibrium}, the alternation of equilibrium point indices along the isoclines extends to include equilibrium points located on the coordinate axes, with the exception of the origin $(0,0)$.
\end{prop}

More precisely, this alternation has to be understood along given connected isoclines $\Gamma_f$ and $\Gamma_g$: consecutive equilibrium points encountered when following an isocline have alternating indices, including possible intersections with the coordinate axes.
This result shows that the index alternation is a global geometric property of the vector field induced by \eqref{eq:init}, and is not restricted to strictly positive coexistence equilibria. 
Equilibrium points located on the axes naturally fit into the same alternation pattern, while the origin $(0,0)$ plays a singular role and is therefore excluded from this statement.
Indeed, the statement concerns the ordering of equilibria along the isoclines. 
It does not imply that two arbitrary boundary equilibria located on different axes, or on different isoclines, must have opposite indices. 
In particular, if the isoclines do not intersect in the strictly positive quadrant, then there is no interior coexistence equilibrium and no alternation sequence connecting the two boundary equilibria through the positive quadrant.

\begin{ex}
    Let us consider the following linear mutualistic system:
    \begin{equation}
\left\{
    \begin{array}{ll}
        \dot{x} &= x f\left(x,y\right) = x\left(a_1 + b_1y - c_1x\right), \\
        \dot{y} &= y g\left(x,y\right) = y\left(a_2 - b_2y + c_2x\right),
    \end{array}
\right.
\label{eq:ex_axes}
\end{equation}
with $\left(a_i, b_i, c_i, i \in \{1,2\} \right)$ being positive constants and $\frac{c_2}{b_2}>\frac{c_1}{b_1}$.
Then, the intersection of $\Gamma_f$ and $\Gamma_g$ lies outside the strictly positive quadrant, and the system has no strictly positive coexistence equilibrium.
However, the system has two boundary equilibria: $\left(0, \frac{a_2}{b_2}\right)$ and $\left(\frac{a_1}{c_1},0\right)$, two saddle points with index $-1$.

The two boundary equilibria are not connected by an alternation sequence along a same isocline through positive coexistence equilibria. 
Rather, the positive-quadrant intersection of the two isoclines is absent.
In this case, there is no index alternation to be followed inside the strictly positive quadrant.
\end{ex}

This extension is important biologically because equilibria located on the axes correspond to single-species persistence states. 
Proposition \ref{prop:equilibrium_axes} therefore shows that such boundary equilibria are constrained by the same index-alternation structure as coexistence equilibria. 
In particular, a stable single-species equilibrium and a stable coexistence equilibrium cannot be the only equilibrium configurations encountered consecutively along the same isocline: an intermediate saddle-type equilibrium must occur between them.

\subsection{Intraspecific Competition and Positive Intraspecific Density Dependence}
\label{subsec:Intra_Competition}

From this section, we assume that the functions $f$ and $g$ are of class $\cun$.
By applying the implicit function theorem and using the fact that the relevant partial derivatives have strict signs locally, we obtain information about the local monotonicity of the $x$-isocline $\Gamma_f$ and the $y$-isocline $\Gamma_g$.
These results are summarized in Table~\ref{tab:eq_iso}.

In many models presented in \cite{hale2021ecological}, intraspecific competition is a key mechanism in regulating population dynamics within each species, preventing the uncontrolled population growth.  
With this additional assumption, all previous equilibrium points with an index of $+1$ become attractive points. 
Under these biologically interpretable assumptions on intraspecific competition, the index-based framework therefore not only constrains the arrangement of equilibria, but also helps determine which coexistence equilibria are locally stable.

\begin{cor}\label{cor:attrac}
    Let a dynamical system in $\R_+ \times \R_+$ be described by \eqref{eq:init}, with functions $f$ and $g$ both $\cun$ satisfying hypotheses \ref{hyp:one} to \ref{hyp:four}, and $\frac{\partial f}{\partial x} < 0$ and $\frac{\partial g}{\partial y} < 0$. 
    Then, in the strictly positive quadrant, the system exhibits an alternation of attractive equilibria and saddle points along the isoclines.
\end{cor}

\begin{remark}
For a model satisfying the conditions of Corollary \ref{cor:attrac}, if there are at least two equilibrium points in the positive quadrant, then an attractive point exists. 
This implies the existence of conditions under which the coexistence of the two species is possible.
\end{remark}

Thus, negative assumptions on the sign of the partial derivatives $\frac{\partial f}{\partial x} < 0$ and $\frac{\partial g}{\partial y} < 0$ (which biologically correspond to intraspecific competition, always negative regardless of species densities) prevent the emergence of any coexistence solutions other than attractive equilibrium points or saddle points.
As in Proposition \ref{prop:equilibrium_axes}, this corollary can be extended to the coordinate axes, where the same monotonicity assumptions guarantee uniqueness and attractivity of equilibria.

\begin{cor}\label{cor:attrac_axes}
Under the assumptions of Corollary \ref{cor:attrac}, if an equilibrium point other than the origin exists on a coordinate axis in $\R_{+}\times\R_{+}$, then it is unique on that axis. 
Moreover, if this equilibrium has index $+1$, then it is an attractive equilibrium.
\end{cor}

Under the opposite constant sign assumptions, \textit{i.e.} intraspecific positive feedbacks for both species as described in Section \ref{subsec:Key_Assumptions_Def}, the system in the positive quadrant alternates between repulsive equilibria and saddle points along the isoclines.

\begin{cor}\label{cor:repul}
    Let a dynamical system in $\R_+ \times \R_+$ be described by \eqref{eq:init}, with functions $f$ and $g$ both $\cun$ satisfying hypotheses \ref{hyp:one} to \ref{hyp:four}, and $\frac{\partial f}{\partial x} > 0$ and $\frac{\partial g}{\partial y} > 0$. 
    Then, in the strictly positive quadrant, the system exhibits an alternation of repulsive equilibria and saddle points along the isoclines.
\end{cor}

This reflects a dynamic where positive intraspecific feedback leads to instability in some regions, while stability is limited to saddle point behavior in others.
Under these assumptions, this reflects a fragile coexistence scenario where small perturbations can push the system away from equilibrium, leading to instability in the dynamics of the two species.

This corollary can also be extended to the coordinate axes: the same monotonicity properties imply that any equilibrium lying on an axis is unique there, and if it has index (+1), it is a repulsive equilibrium.
\begin{cor}\label{cor:repul_axes}
Under the assumptions of Corollary \ref{cor:repul}, if an equilibrium point exists on a coordinate axis in $\R_{+}\times\R_{+}$, then it is unique on that axis. 
Moreover, if this equilibrium has index $+1$, then it is a repulsive equilibrium.
\end{cor}

\subsection{Illustration of Theorem \ref{thm:equilibrium} with Mutualistic Models from the Literature}
\label{examples}

In this section, we illustrate Theorem \ref{thm:equilibrium} using four models of mutualistic interactions reviewed in \cite{hale2021ecological} (see Table \ref{tab:saturing} and Table \ref{tab:non-mono} in Appendix~\ref{app:mutualism_models}), each represented on Figure \ref{fig:global} (in the order presented), each emphasizing a different characteristic of mutualism.
Although their mathematical formulation reflects different perspectives on mutualistic interactions, as discussed in Section \ref{subsec:Key_Assumptions_Def}, we show that they satisfy the assumptions of system \eqref{eq:init}, highlighting the generality of the framework provided in our study.

\begin{itemize}
    \item[(a)] In \cite{zhang2003mutualism} the author models the dynamics of two species interacting along a mutualistic-parasitic continuum using the same mathematical formulation twice, reflecting an interaction between species at the same trophic level.
    In this model, the growth rate of each species is influenced by a quadratic term that depends on the density of the other species.
    The intraspecific competition is represented by a constant term, while the interaction between species varies in sign: it remains positive up to a threshold, beyond which it becomes negative. 
    This formulation captures the effect of saturation and illustrates the continuum between mutualism and parasitism.
    Here, there is the assumption of intraspecific competition, where $\frac{\partial f}{\partial x} < 0$ and $\frac{\partial g}{\partial y} < 0$.
    Moreover, the interaction is mutualistic as long as $y$ remains below the threshold $b_1$, but when $y$ exceeds this value, $\frac{\partial f}{\partial y}$ turns negative. 
    This dynamic highlights how the strength and nature of the interaction shift depending on species densities.

    \item[(b)] \cite{neuhauser2004mutualism} introduce a model describing facultative mutualism between plants ($x$) and mycorrhizal fungi ($y$). 
    In this formulation, the effects on species $y$ remain constant, regardless of $x$ or $y$, with intraspecific competition always negative $\left(\frac{\partial g}{\partial y} < 0\right)$ and the mutualistic benefit from $x$ always positive $\left(\frac{\partial g}{\partial x} > 0\right)$.
    Similarly, intraspecific competition for $x$ is constant and negative $\left(\frac{\partial f}{\partial x} < 0\right)$. 
    However, the term $\frac{\partial f}{\partial y}$, which characterizes the nature of the interaction, varies depending on the densities of both species.  Specifically, as $y$ increases, the benefit for $x$ diminishes and can eventually turn negative, shifting the interaction away from mutualism. 
    However, this effect is modulated by the presence of $x$: for large values of $x$, $\frac{\partial f}{\partial y}$ is more strongly positive, reinforcing mutualistic benefits.
    $x$ and $y$ must increase together to prevent the relationship from becoming parasitic.

    \item[(c)]The model proposed by \cite{graves2006bifurcation} represents obligate mutualism for lichens, where the benefits of the interaction are tightly linked to the densities of both species. 
    The system assumes constant negative intraspecific competition, with identical functional forms for both species, and the mutualistic interaction varies with the density of the other species.
    Specifically, as the density of $y$ increases, the benefit for species $x$ increases but saturates, since the marginal effect of $y$ decreases with its density. 
    This exponential term prevents unbounded growth of the mutualistic benefit.

    \item[(d)]A more general consumer-resource approach is presented by \cite{nathaniel2009consumer, holland2010}, which models the bidirectional exchange of resources. 
    In this framework, one species provides resources while the other enhances nutrient uptake efficiency. 
    The system follows a similar functional form for both species, but the dynamics are more complex as all aspects, including the intraspecific competition, vary with the densities of both species.
    In this model, the mutualistic interaction is obligate when the intrinsic growth rates are negative ($r_1 \leq 0$, $r_2 \leq 0$).
    The interaction structure is dynamic, with the intraspecific competition for $x$ decreasing as $y$ grows, but increasing as $x$ increases. 
    Conversely, the effect of $y$ on $x$ increases with $x$, but decreases as $y$ increases. 
    This creates a more flexible framework for describing how species' interactions evolve with changing densities. 
\end{itemize}

Depending on the biological aspects these models aim to represent, the mathematical forms differ, leading to a varying number of equilibrium points. 
However, all of them satisfy conditions \ref{hyp:one} to \ref{hyp:four} and exhibit an alternation of equilibrium point indices along the isoclines, as a result leading to one or more stable coexistence configurations.

\begin{figure}[h!]
    \centering
    \begin{subfigure}{0.45\textwidth}
        \centering
        \includegraphics[width=\textwidth]{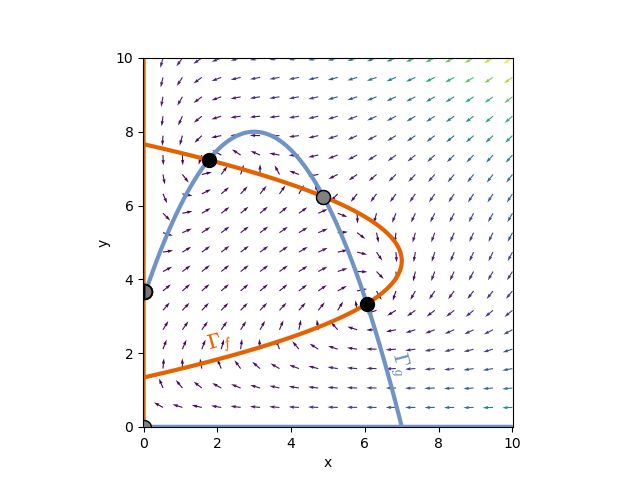}
        \caption{Mutualism at the same trophic level.}
        \label{fig:Zhang}
    \end{subfigure}
    \hfill
    \begin{subfigure}{0.45\textwidth}
        \centering
        \includegraphics[width=\textwidth]{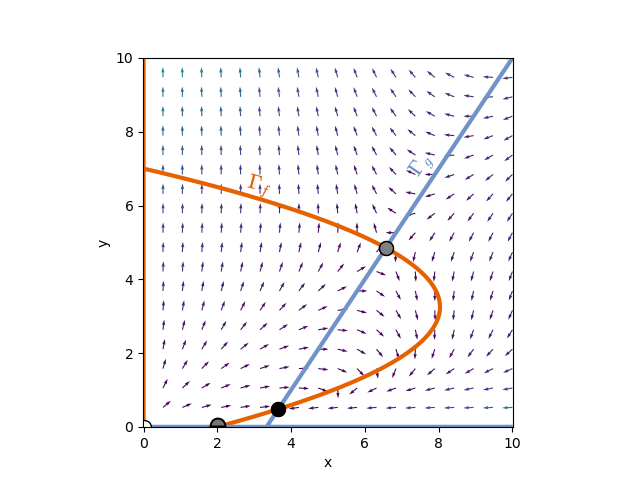}
        \caption{Plant–mycorrhiza, facultative mutualism.}
        \label{fig:Neuhauser}
    \end{subfigure}
    
    \vspace{0.5cm}  

    \begin{subfigure}{0.45\textwidth}
        \centering
        \includegraphics[width=\textwidth]{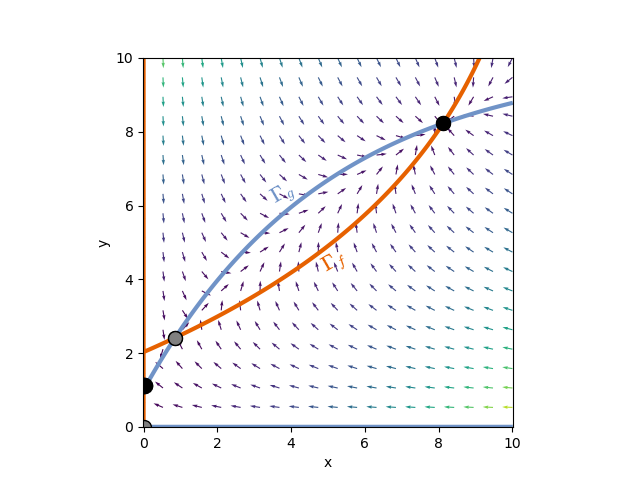}
        \caption{Obligate mutualism in lichens.}
        \label{fig:Graves}
    \end{subfigure}
    \hfill
    \begin{subfigure}{0.45\textwidth}
        \centering
        \includegraphics[width=\textwidth]{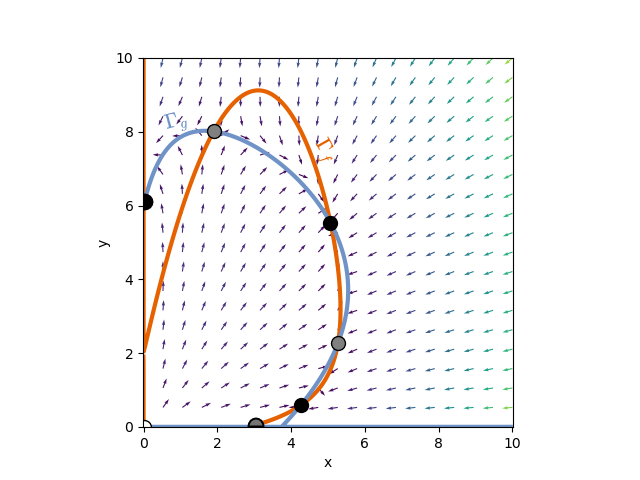}
        \caption{Consumer-resource approach.}
        \label{fig:Holland}
    \end{subfigure}    
    \caption{\textit{Phase portraits of four different models of mutualism.} Each model highlights different aspects of mutualism while satisfying the assumptions \ref{hyp:one} to \ref{hyp:four} of system \eqref{eq:init} and exhibiting an alternation of equilibrium point indices along the isoclines.
    The curve $\Gamma_g$ is shown in blue, and $\Gamma_f$ in orange.
    White points represent repulsive points, black points represent attractive points and gray points represent saddle points.
    (\ref{fig:Zhang}) Parameters for \cite{zhang2003mutualism} model: $c_1 = 7$, $c_2 = 8$, $a_1 = 0.7$, $a_2 = 0.5$, $b_1 = 4.5$, $b_2 = 3$, $r_1 = 7$, $r_2 = 4$.
    (\ref{fig:Neuhauser}) For \cite{neuhauser2004mutualism} model: $r_1 = 1.4$, $r_2 = 1$, $K_1 = 2$, $K_2 = 5$, $\alpha_{21} = 4$, $\alpha_{12} = 1.5$, $a = 0.2$.
    (\ref{fig:Graves}) For \cite{graves2006bifurcation} model: $r_{10} = -2$, $r_{20} = 0.5$, $r_{11} = 4$, $r_{21} = 5$, $k_1 = 0.2$, $k_2 = 0.2$, $a_1 = 0.35$, $a_2 = 0.5$.
    (\ref{fig:Holland}) For \cite{holland2010} model: $r_1 = 2$, $r_2 = 3$, $c_1 = 3$, $c_2 = 1$, $a_{12} = 1$, $a_{21} = 2.6$, $h_1 = 1$, $h_2 = 1$, $q_1 = 1$, $q_2 = 2.7$, $\beta_{12} = 1$, $\beta_{21} = 1.5$, $e_1 = 0.5$, $e_2 = 3$, $s_1 = 0.7$, $s_2 = 0.5$.}
    \label{fig:global}
\end{figure}
\newpage
\section{Periodic solutions in the dynamics of two mutualistic species}
\label{sec:Mutualistic}

In this section, we clarify the role of extended mutualism (Definition \ref{def:mut}) and contrast it with the more constrained definition of strict mutualism.
Specifically, we demonstrate that under the strict mutualism assumptions, where both $\frac{\partial f}{\partial y}>0$ and $\frac{\partial g}{\partial x}>0$ in the same region of the positive quadrant, the system cannot give rise to periodic solutions such as limit cycles in its neighborhood.

\begin{remark}
Here we aim at investigating attractor limit cycles. 
A limit cycle is an isolated closed trajectory that can be stable, unstable or half-stable for nearby trajectories. 
When the limit cycle is stable, trajectories that start close to it tend to converge towards it over time, which stabilizes the oscillation. 
Thus, attractive limit cycles represent stable periodic oscillations, where populations repeatedly return to the same cycle of variation.
In mutualistic interactions, the presence of limit cycles might suggest that the populations of both species oscillate over time in a stable way, without extinction or explosive growth.
In what follows, we mainly focus on attractive limit cycles, which correspond to stable periodic oscillations. When needed, we explicitly mention repulsive limit cycles.
\end{remark}

\begin{remark}
Limit cycles differ from centers, such as those observed in the Lotka-Volterra predator-prey system.
Centers are equilibrium points where the trajectories of the system form closed, periodic orbits, but they are neither attractive nor repulsive. 
As a result, the behavior of the system is highly sensitive to initial conditions: each new set of initial conditions produces a distinct trajectory, making the system highly sensitive to small variations \citep{may1972}. 
However, even in the case of centers, when the Lotka-Volterra predator-prey model is modified to represent mutualistic interactions, the centers disappear. 
\end{remark}

We recall that periodic solutions can only surround a region with an index of $+1$ (Figure \ref{fig:PhasePortrait_Index+1}).
However, equilibrium points do not necessarily generate an attractive limit cycle around them. 
Limit cycle will arise around a repulsive equilibrium point, where nearby trajectories diverge from the equilibrium, but are eventually captured by the limit cycle.

\subsection{Case of strict mutualism}
\label{subsec:Strict_Mutualism}

We first study whether a limit cycle can be observed under the assumption that the repulsive equilibrium point is strictly mutualistic, that is, under the assumption $\frac{\partial f}{\partial y} > 0$ and $\frac{\partial g}{\partial x} > 0$.

\begin{thm}
\label{thm:no_cycle_lim_1}
    Let a dynamical system in $\R_+ \times \R_+$ be described by \eqref{eq:init}, with functions $f$ and $g$ of class $\cun$ satisfying hypotheses \ref{hyp:one} to \ref{hyp:four}.
    We assume that the curves $\Gamma_f$ and $\Gamma_g$ intersect at a repulsive equilibrium point $\left(x^*, y^* \right)$, with $\frac{\partial f}{\partial y} (x^*, y^*) > 0$ and $\frac{\partial g}{\partial x} (x^*, y^*) > 0$. 
    Then, there exists a region $R_1$ in the phase portrait containing this point, in which the partial derivatives do not change sign, and no limit cycle exists that surrounds this point within $R_1$.
\end{thm}

We next consider the case of an attractive equilibrium that could be surrounded by a repulsive limit cycle.
The following theorem shows that, under strict mutualism and constant signs of the relevant partial derivatives, such a repulsive limit cycle cannot occur in a neighborhood of the equilibrium.

\begin{thm}
\label{thm:no_cycle_lim_2}
    Let a dynamical system in $\R_+ \times \R_+$ be described by \eqref{eq:init}, with functions $f$ and $g$ of class $\cun$ satisfying hypotheses \ref{hyp:one} to \ref{hyp:four}.
    We assume that the curves $\Gamma_f$ and $\Gamma_g$ intersect at an attractive equilibrium point $\left(x^*, y^* \right)$, with $\frac{\partial f}{\partial y} (x^*, y^*) > 0$ and $\frac{\partial g}{\partial x} (x^*, y^*) > 0$. 
    Then, there exists a region $R_2$ in the phase portrait containing this point, in which the partial derivatives do not change sign, and no limit cycle exists that surrounds this point within $R_2$.
\end{thm}

Thus, around a point with index +1 (attractive or repulsive), in a system \eqref{eq:init} where the conditions on the partial derivatives do not change, we cannot find periodic solutions such as limit cycles.
As the Definition \ref{def:mut} of extended mutualism includes the case of strict mutualism, we deduce that in the region where both species have a strictly positive effect on each other and in which the partial derivatives do not change sign, no periodic solutions such as limit cycles can arise.
In particular, strict mutualism alone does not provide the alternation between positive and negative feedbacks that is a classical feature of periodic dynamics in  two-species systems.
This does not mean that a trajectory in a density-dependent mutualistic model can never pass through a strictly mutualistic region.
Rather, it means that a limit cycle cannot be generated solely by a region in which the interaction remains strictly mutualistic and the signs of the partial derivatives remain unchanged.
We will now explore what happens outside this region.

\subsection{Case of extended mutualism}
\label{subsec:Extended_Mutualism}

In this section, we use the extended definition of mutualism provided in Definition \ref{def:mut}, which includes the previous case of strict mutualism as well as the situation where the  interactions between species can be mutualistic only for certain combinations of species densities.
As previously discussed, the emergence of a limit cycle can indeed only occur in a context where the interaction is either parasitic or competitive, and not in a region where both $\frac{\partial f}{\partial y}>0$ and $\frac{\partial g}{\partial x}>0$ simultaneously and remain unchanged. 
Here, we consider density-dependent effects that give rise to regions where the mutualistic interaction locally shifts to a parasitic one.
We identify two different scenarios where the nature of the interaction is not fixed, i.e. where at least one of the cross derivatives $\frac{\partial f}{\partial y}$ or $\frac{\partial g}{\partial x}$ changes sign, leading to the emergence of a limit cycle.
In a related but more specific modeling framework, \cite{song2025} described another example of limit cycle dynamics, based on a detailed mutualism model in which variations in interaction costs and benefits are explicitly analyzed, and shown to induce or suppress cyclic behavior.

We also introduce the assumption that intraspecific competition can change sign and become positive intraspecific density-dependence (for example, $\frac{\partial g}{\partial y}$ may change sign). 
In mutualism, this can occur if an increase in the population of species $x$ leads to a higher demand for individuals of species $y$ to support the interaction.
To investigate this scenario, we focus on a single equilibrium point in the population quadrant, excluding the one at $(0,0)$ and those on the axes. 

\paragraph{}For our first configuration, we make the following assumptions: 

\begin{enumerate}[label=1.\arabic*]
    \item $\frac{\partial f}{\partial x}<0$: Intraspecific competition, the density of species $x$ is always limited by the intraspecific competition among individuals.
    \item $\frac{\partial g}{\partial y}>0$ then $\frac{\partial g}{\partial y}<0$: At low density, species $y$ experiences positive intraspecific density dependence, corresponding to a weak Allee effect or, more generally, to an Allee-type positive density dependence. 
    At higher density, negative density dependence dominates, reflecting intraspecific limitation.
    \item $\frac{\partial g}{\partial x}>0$: Species $x$ always has a positive effect on the density of species $y$.
    \item $\frac{\partial f}{\partial y}>0$ then $\frac{\partial f}{\partial y}<0$: At low density, species $y$ has a positive effect on species $x$, but as its density increases, the interaction shifts from mutualism to parasitism. 
    \item $f\left(0,y_1\right)=0$: Above the threshold $y_1$, the population of $x$ declines, regardless of whether its own density is low or high.
    \item $g\left(x_1,0\right)=0$:  
    The threshold $x_1$ of species $x$ required for species $y$ to persist at low density. 
    \item $f\left(x_2,0\right)=0$: There exists an equilibrium at $x_2$ where species $x$ can persist in the absence of species $y$.
    Beyond $x_2$, the species $x$ declines due to overpopulation.
    \item $x_1<x_2$: The threshold density for species $x$ to persist in isolation ($x_2$) is higher than the threshold where species $x$ can sustain $y$ ($x_1$). 
    Otherwise, species $y$ goes extinct.
    \item $\left(x^*, y^*\right)$ is a repulsive equilibrium point: 
    \begin{equation*}
        x^* \frac{\partial f}{\partial x} \left(x^*,y^*\right)+y^* \frac{\partial g}{\partial y}\left(x^*,y^*\right)>0,
    \end{equation*}
    and
    \begin{equation*}
        x^* y^* \left( \frac{\partial f}{\partial x} \left(x^*,y^*\right) \frac{\partial g}{\partial y}\left(x^*,y^*\right) - \frac{\partial f}{\partial y} \left(x^*,y^*\right) \frac{\partial g}{\partial x}\left(x^*,y^*\right) \right) >0.
    \end{equation*}
    This condition ensures that small perturbations around $(x^*, y^*)$ will lead to divergence, favoring oscillatory or cyclic behavior. 
    \item The isoclines $f=0$ and $g=0$ have the shapes illustrated in Figure \ref{fig:cycle_lim_1}.
    \begin{figure}[h!]
    \centering
    \includegraphics[width=0.46\linewidth]{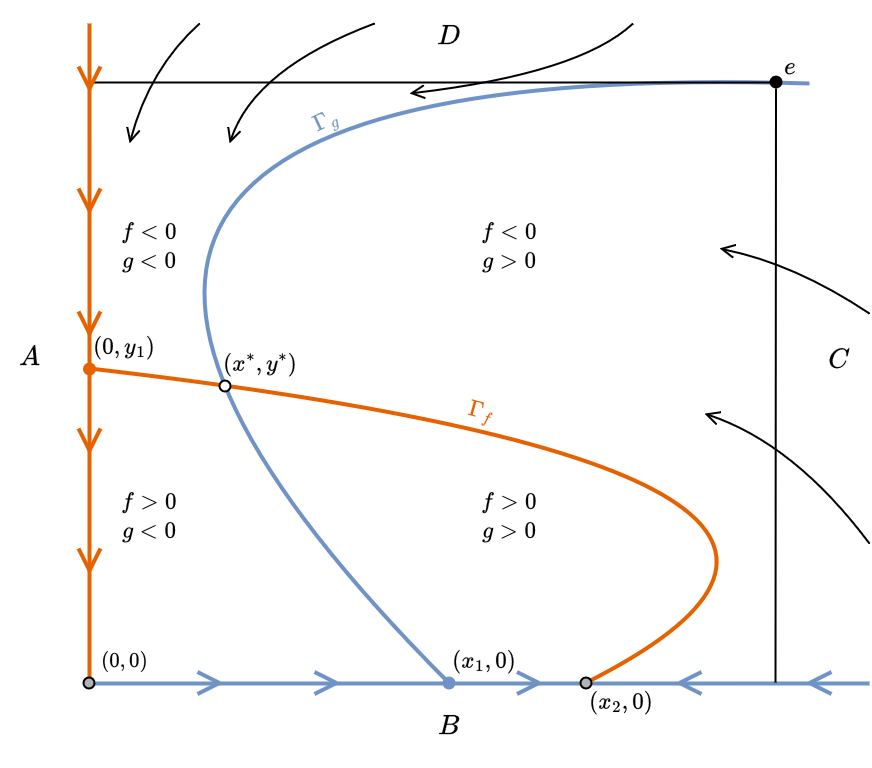}
    \caption{\textit{Phase portrait assuming extended mutualism leading to a cyclic behaviour (illustration of 1.10 of the rectangular region $ABCD$).}  
    The curve $\Gamma_g$ is shown in blue, and $\Gamma_f$ in orange.
    By assumption 1.9, $(x^*, y^*)$ is a repulsive point (the white point), which implies, by Theorem \ref{thm:equilibrium}, that $(x_2,0)$ is a saddle point (the grey point), as well as $(0,0)$.}
    \label{fig:cycle_lim_1}
    \end{figure}
\end{enumerate}

\newpage
For our second configuration, we make the following assumptions: 

\begin{enumerate}[label=2.\arabic*]
    \item $\frac{\partial g}{\partial y}<0$: Intraspecific competition for species $y$.
    \item $\frac{\partial f}{\partial x}>0$ then $\frac{\partial f}{\partial x}<0$: At low density, species $x$ experiences positive intraspecific density dependence, corresponding to a weak Allee effect or, more generally, to an Allee-type positive density dependence.  
    At higher density, negative density dependence dominates, reflecting intraspecific limitation.
    \item $\frac{\partial f}{\partial y}>0$: Species $y$ always has a positive effect on species $x$.
    \item $\frac{\partial g}{\partial x}<0$ then $\frac{\partial g}{\partial x}>0$: At low density, species $x$ has a negative effect on species $y$, but as its density increases, the interaction shifts from parasitism to mutualism. 
    \item $f(0, y_1) = 0$: Beyond the threshold $y_1$, the population of $x$ grows due to mutualism, regardless of its own density.
    \item $g\left(0,y_2\right)=0$: There exists an equilibrium at $y_2$ where species $y$ can persist in the absence of species $x$.
    Beyond $y_2$, the species $y$ declines due to overpopulation.
    \item $y_1<y_2$: The threshold density for species $y$ to persist in isolation ($y_2$) is higher than the threshold where species $y$ can sustain $x$ ($y_1$). 
    Otherwise, species $x$ goes extinct.
    \item $\left(x^*, y^*\right)$ is a repulsive equilibrium point.
    \item The isoclines $f=0$ and $g=0$ have the shapes illustrated in Figure \ref{fig:cycle_lim_2}.
    \begin{figure}[h!]
    \centering
    \includegraphics[width=0.46\linewidth]{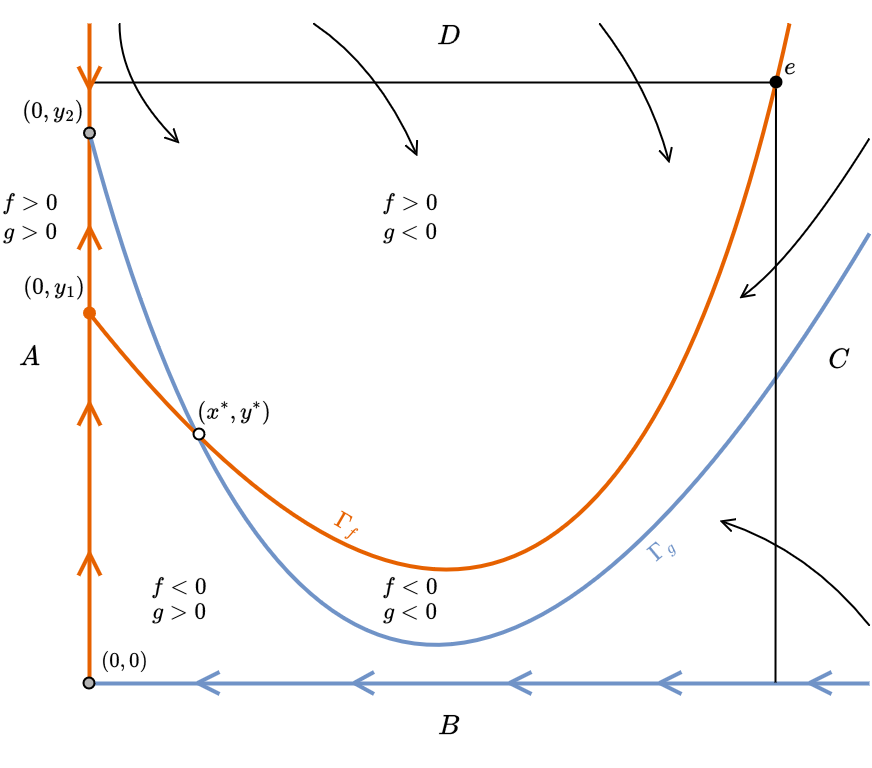}
    \caption{\textit{Phase portrait assuming partial mutualism leading to a cyclic behaviour (illustration of 2.9 of the rectangular region $ABCD$).}
    By assumption 2.8, $(x^*, y^*)$ is a repulsive point (white point), which implies that $(0,0)$ and $(0,y_2)$ are saddle points (grey points).
    The curve $\Gamma_g$ is shown in blue, and $\Gamma_f$ in orange.}
    \label{fig:cycle_lim_2}
    \end{figure}
\end{enumerate}

\newpage
\begin{remark}
Assumptions 1.1-1.10 and 2.1-2.9 are stronger than those in Section \ref{subsec:General_Assumptions} and imply assumptions \ref{hyp:one} to \ref{hyp:four}.
\end{remark}

Under conditions 1.1-1.10 or 2.1-2.9, we have the following result, which will be proven in Section \ref{sec:Proofs}.

\begin{prop}\label{thm:cycle_lim_1}
Let a $\cun$ dynamical system on $\mathbb{R}_+ \times \mathbb{R}_+$ be described by~\eqref{eq:init}, with functions $f$ and $g$ satisfying conditions 1.1-1.10 or 2.1-2.9. 
Then a limit cycle exists inside the strictly positive quadrant.
\end{prop}

We can compare the first configuration for example, with classical predator--prey cycles. 
In the region where $\frac{\partial f}{\partial y}<0$ and $\frac{\partial g}{\partial x}>0$, species $y$ has a negative effect on the growth of species $x$, while species $x$ has a positive effect on the growth of
species $y$. 
Locally, this sign structure is therefore analogous to a
predator--prey or parasitic interaction. 
However, we will see in Example \ref{ex:ex_cycle} that the cyclic dynamics in Figure \ref{fig:cycle_lim_1} are not confined to such an antagonistic region. 
Since $\frac{\partial g}{\partial x}>0$ throughout the configuration, and
$\frac{\partial f}{\partial y}$ is positive at lower densities of species $y$, the trajectory also passes through a region of strict mutualism, where both cross effects are positive.
Thus, although the repulsive equilibrium $(x^\ast,y^\ast)$ is not located in a strictly mutualistic region, the limiting cycle itself crosses such a region.

In a classical predator--prey cycle, oscillations are driven by a fixed antagonistic interaction: the prey benefits the predator, while the predator harms the prey. 
In the present configuration, the interaction changes along the cycle. 
Part of the trajectory corresponds to a mutualistic phase, in which the two species locally benefit from each other, whereas another part corresponds to a parasitic or predator--prey-like phase, in which species $y$ negatively affects species $x$.

\begin{ex}
\label{ex:ex_cycle}
Let us consider the following system:  

\begin{equation}
\left\{
    \begin{array}{ll}
        \dot{x} &= x f\left(x,y\right) = x\left(a_1-b_1\left(y - c_1\right)^{2} - d_1 x\right), \\
        \dot{y} &= y g\left(x,y\right) = y\left(-a_2-b_2\left(y - c_2\right)^{2} + d_2 x\right),
    \end{array}
\right.
\label{eq:ex_cycle}
\end{equation}
with $(a_i, b_i, c_i, d_i, i \in \{1,2\})$ being positive constants.
For a certain range of parameters, the functions $f$ and $g$ will satisfy assumptions 1.1 to 1.10. 
However, according to assumptions 1.9 and 1.10, $\Gamma_f$ and $\Gamma_g$ must intersect while both being decreasing in order to obtain a repulsive equilibrium point.  
Lowering the threshold $c_2$, increasing the threshold $c_1$, or decreasing the interaction strength $b_1$ leads to a shift in the equilibrium location, making it attractive instead.
System \eqref{eq:ex_cycle} is illustrated in Figure \ref{fig:cycle_lim_illus}.

\begin{figure}[h!]
    \centering
    \includegraphics[width=0.7\linewidth]{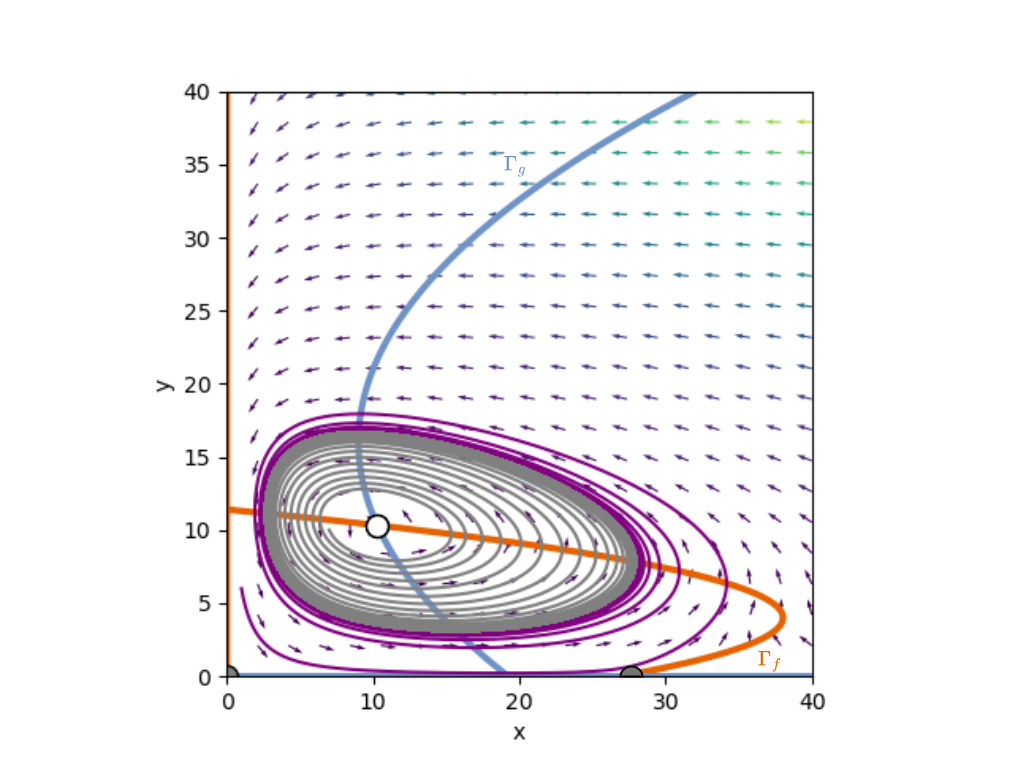}
    \caption{\textit{Phase portrait leading to a limit cycle.}
    Simulation of system \eqref{eq:ex_cycle}  with parameters: $a_1=13.69$, $b_1=0.25$, $c_1=4$, $d_1=0.36$ and $a_2=9$, $b_2=0.04$, $c_2=16$, $d_2=1$.
    The curve $\Gamma_g$ is shown in blue, and $\Gamma_f$ in orange.
    The white point represents a repulsive point and gray points represent saddle points.
    The purple trajectory starts at a population density of 2 for species $x$ and 6 for species $y$. 
    It converges toward the limit cycle from the outside, moving away from the extinction states, either of both species or the extinction of $y$.
    The gray trajectory starts at densities of 7 for species $x$ and 10 for species $y$, and converges toward the limit cycle from the inside, moving away from the interior equilibrium point.}
    \label{fig:cycle_lim_illus}
    \end{figure}
    
With the system \eqref{eq:ex_cycle}, although the repulsive equilibrium is not located in a strictly mutualistic region, the limiting trajectory crosses the region where $\frac{\partial f}{\partial y}>0$ and $\frac{\partial g}{\partial x}>0$, corresponding to strict mutualism. 
Hence, the cycle alternates between a strictly mutualistic phase and a parasitic, predator--prey-like phase.
\end{ex}

\newpage
In both models presented above, Assumption 1.2 (in the first configuration) and Assumption 2.2 (in the second configuration) impose a change of sign in the intraspecific interaction term. 
More precisely, for at least one species, the partial derivative of its per capita growth rate with respect to its own density satisfies in the first configuration $\frac{\partial g}{\partial y} > 0$ at low density and $\frac{\partial g}{\partial y} < 0$ at high density, and analogously for $f$ in the second configuration.
Biologically, this required intraspecific sign change takes the form of a weak Allee effect, since the per-capita growth rate increases with intraspecific density at low density. 
However, the assumptions do not necessarily imply a strong Allee effect, which would require an additional condition on the sign of the per-capita growth rate itself ($f(x,y)$ or $g(x,y)$), not only on the sign of its derivative.
This sign change induces a non-monotonic self-regulation mechanism. 
Our analysis suggests that such a mechanism is structurally linked to the existence of a limit cycle.
{In particular, Proposition~4.8 below excludes limit cycles when the two intraspecific terms keep the same constant sign throughout the positive quadrant. 
The cyclic configurations constructed here avoid this criterion because one intraspecific derivative changes sign.

\begin{prop}\label{thm:cycle_lim_2}
Let a $\cun$ dynamical system on $\mathbb{R}_+ \times \mathbb{R}_+$ be described by~\eqref{eq:init}.
Assume that the functions $\frac{\partial f}{\partial x}$ and $\frac{\partial g}{\partial y}$ are both non-positive everywhere, or both non-negative everywhere, and are not both identically zero in the strictly positive quadrant. 
Then the system does not admit any limit cycle inside the strictly positive quadrant.
\end{prop}

In particular, Proposition \ref{thm:cycle_lim_2} excludes limit cycles when the two intraspecific density-dependent terms have the same constant sign throughout the positive quadrant.
The cyclic examples constructed above in Section~\ref{subsec:Extended_Mutualism} escape this criterion because one intraspecific term changes sign.
Biologically, in these examples, this sign change corresponds to a weak Allee effect, but the proposition itself is not restricted to Allee effects.
However, Proposition~4.8 does not exclude all possible systems with constant intraspecific signs of opposite sign.

The result, proved in Section \ref{sec:Proofs}, follows from a variant of the Dulac--Bendixson criterion (see \cite{perko2013differential}).

\subsection{From Parasitism to Mutualism: Parameter-Driven Transitions in Interaction Dynamics}

In this section, we illustrate how our framework can be applied to models that cannot be captured by a definition of strict mutualism. 
We propose two models that, while not always strictly mutualistic, satisfy Definition \ref{def:mut}, and fit to our general framework.
Depending on population densities and on the location of the equilibrium in the phase portrait, such models may represent parasitic, competitive, or mutualistic interactions in different regions.
In particular, Definition \ref{def:mut} allows for systems in which the nature of species interactions may change continuously as parameters vary, while still remaining within the same mathematical structure.
Importantly, the general results established in Theorem \ref{thm:equilibrium} apply independently of transitions between mutualism and parasitism, as the theorem is formulated under assumptions that are broader than those imposed in Section \ref{subsec:Key_Assumptions_Def}.
As a consequence, a model within this framework may describe qualitatively different ecological interactions solely through continuous changes in parameter values—such as the strength or sign of intraspecific competition—without leaving the scope of Theorem \ref{thm:equilibrium}.

First, we examine a model which, depending on the signs of its coefficients—and thus its partial derivatives—can represent either a parasitic or a mutualistic interaction, as discussed in Section \ref{subsec:Key_Assumptions_Def}.  
We analyze the consequences of this change on phase portraits and equilibrium points of index $+1$.
This first example corresponds to a change in the sign structure of the interaction terms. 
Such a change moves the system from one qualitative interaction regime to another and may reverse the stability of equilibria of index $+1$.
We then consider a second model that falls within Definition \ref{def:mut}, where we keep the coefficient signs unchanged but slightly modify the values of certain parameters.
We observe that an equilibrium previously identified as attractive in a parasitic region remains attractive but now in a mutualistic region for both species.
This second example preserves the sign structure of the partial derivatives but moves the equilibria within the phase portrait. 
Thus, the same mathematical model may lead to different ecological interpretations depending on the region in which the equilibrium is located.
The following examples are consistent with assumptions \ref{hyp:one} to \ref{hyp:four}. 

\paragraph{}Our first example is the following system:  

\begin{equation}
\label{eq:Para_Mut}
    \begin{cases}
     \dot{x}= \delta x \left(a_1- b_1 \left(x-c_1 \right)^2- d_1 y \right), \\
     \dot{y}= \delta y \left(a_2-b_2 \left(y-c_2 \right)^4- d_2 x \right),
\end{cases}
\end{equation}
where:  
\begin{equation*}
\delta =
\begin{cases}
1 & \text{for parasitism} \\
-1 & \text{for mutualism}.
\end{cases}
\end{equation*}

Depending on the sign of $\delta$, the same model can represent either a parasitic interaction, where $\frac{\partial g}{\partial x}<0$ and $\frac{\partial f}{\partial y}<0$, or a mutualistic interaction, where both $\frac{\partial g}{\partial x}>0$ and $\frac{\partial f}{\partial y}>0$.  

Figure \ref{fig:Ex_Para_Mut} presents two phase portraits of this model, illustrating the effect of changing the sign of $\delta$.  
This modification transforms the system from a parasitic interaction—where a stable equilibrium allows both species to coexist, two additional equilibria enable either species to outcompete the other and persist alone, and extinction at $(0,0)$ is an attractor—to a mutualistic interaction.  
In the mutualistic case, a stable coexistence equilibrium remains, but at lower population densities. 
Extinction becomes repulsive, and if initial densities are sufficiently high, population sizes can grow explosively.  

While the alternation predicted by the theorem remains unchanged, the nature of equilibrium points of index $+1$ shifts: previously attractive points become repulsive, and vice versa.
This example shows that if for any reason the nature of the interaction between two species suddenly switches without modifying the strength of the interactions of their quantitative effects on birth and death rates, the dynamics of the two species community can dramatically change.

\begin{figure}[h]
    \centering
    \begin{subfigure}[b]{0.45\textwidth}
        \centering
        \includegraphics[width=\textwidth]{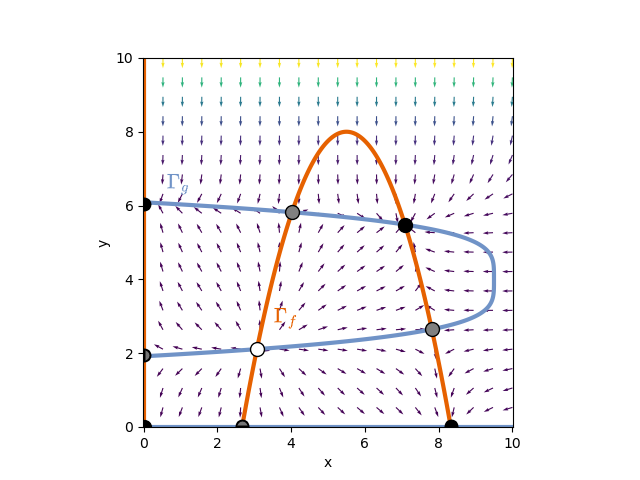}
        \caption{Phase portrait of parasitism of \eqref{eq:Para_Mut}, $\delta=1$.}
        \label{fig:Ex_Parasitism}
    \end{subfigure}
    \hfill
    \begin{subfigure}[b]{0.45\textwidth}
        \centering
        \includegraphics[width=\textwidth]{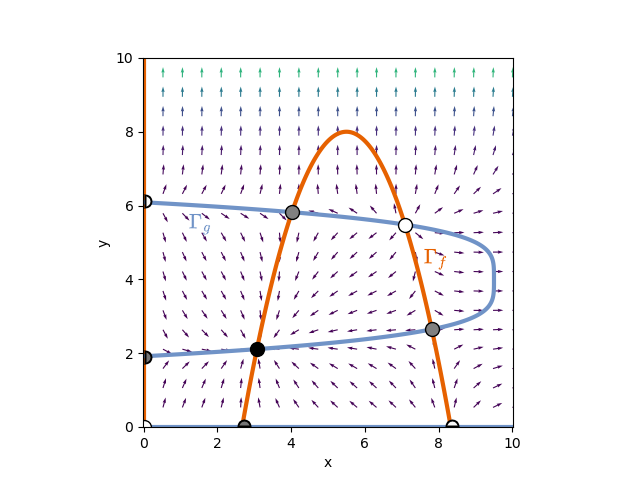}
        \caption{Phase portrait of mutualism of \eqref{eq:Para_Mut}, $\delta=-1$.}
        \label{fig:Ex_mutualism}
    \end{subfigure}
    \caption{\textit{Phase portraits of system \eqref{eq:Para_Mut} illustrating how changes in parameter signs alter the interaction type.} 
    The curve $\Gamma_g$ is shown in blue, and $\Gamma_f$ in orange.
    White points represent repulsive points, black points represent attractive points and gray points represent saddle points.
    With: $a_1=8$, $b_1=1$, $c_1=5.5$, $d_1=1$ and $a_2=9.5$, $b_2=0.5$, $c_2=4$, $d_2=1$. }
    \label{fig:Ex_Para_Mut}
\end{figure}

\newpage 

In our second example, we do not change the sign of the partial derivatives; the model has two regions, a mutualistic region and a parasitic region. We maintain the nature of the interaction but slightly modify the values of certain constants, which shifts the positions of the equilibrium points.
We use the following model:  

\begin{equation}
\label{eq:Para_Mut2.0}
    \begin{cases}
    \dot{x} = x(-a_1 - b_1 (x - c_1)^2 + d_1 y), \\
    \dot{y} = y (a_2 - b_2 (x - c_2)^4 - d_2 y).
\end{cases}
\end{equation}

Figure \ref{fig:Ex_Para_Mut_2.0} provides two phase portraits of the same model, demonstrating how varying parameter values affect equilibrium locations.  
The growth region of $\Gamma_g$ corresponds to $\frac{\partial g}{\partial x} > 0$, while its decline region corresponds to $\frac{\partial g}{\partial x} < 0$, with $\frac{\partial g}{\partial y} < 0$ always.  
A symmetric structure applies to $f$.  

Thus, in Figure \ref{fig:Ex_Parasitism_2.0}, the stable equilibrium lies in a region where the interaction is parasitic for species $y$ and beneficial for species $x$.  
In contrast, in Figure \ref{fig:Ex_mutualism_2.0}, stable coexistence occurs within a region of strict mutualism.
This second example shows that the nature of interactions between two species  at ecological equilibrium can dramatically change if the intensity of the interactions changes, even though the explicit form of the density-dependent functions $f$ and $g$ does not change.

\begin{figure}[h]
    \centering
    \begin{subfigure}[b]{0.45\textwidth}
        \centering
        \includegraphics[width=\textwidth]{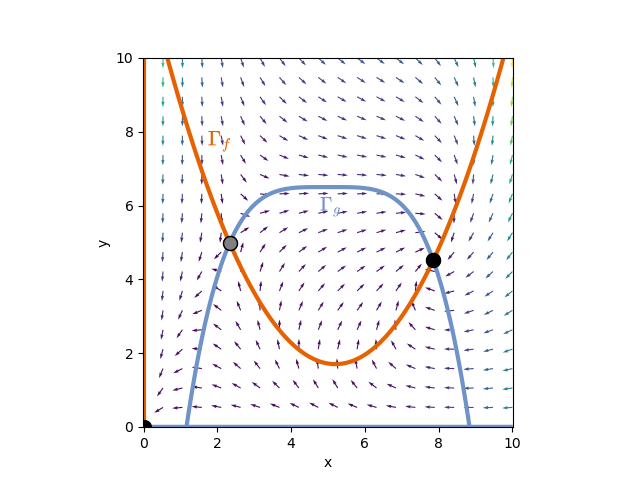}
        \caption{Phase portrait of \eqref{eq:Para_Mut2.0}, with parameters: $a_1=1.7$, $c_1=5.2$ and $c_2=5$.}
        \label{fig:Ex_Parasitism_2.0}
    \end{subfigure}
    \hfill
    \begin{subfigure}[b]{0.45\textwidth}
        \centering
        \includegraphics[width=\textwidth]{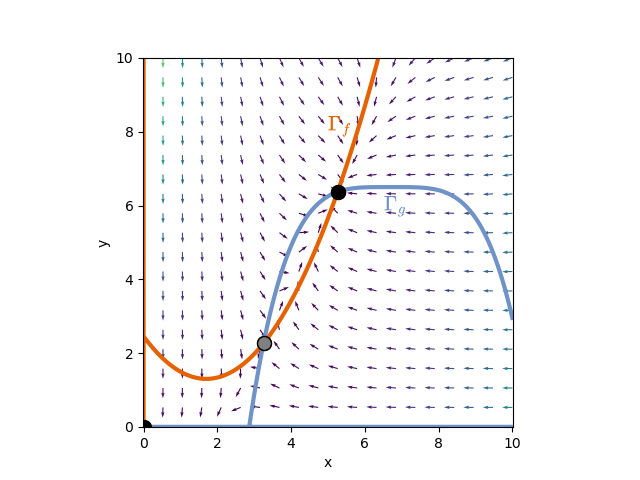}
        \caption{Phase portrait of \eqref{eq:Para_Mut2.0}, with parameters: $a_1=1.5$, $c_1=1.9$ and $c_2=7$.}
        \label{fig:Ex_mutualism_2.0}
    \end{subfigure}
    \caption{\textit{Phase portraits illustrating how changing parameter values—while maintaining the interaction type—modifies the equilibrium locations.}
    The curve $\Gamma_g$ is shown in blue, and $\Gamma_f$ in orange.
    Black points represent attractive points and gray points represent saddle points.
    System \eqref{eq:Para_Mut2.0} with: $b_1=0.4$, $d_1=1$, $a_2=6.5$, $b_2=0.03$ and $d_2=1$.}
    \label{fig:Ex_Para_Mut_2.0}
\end{figure}
\newpage
\section{Discussion}
\label{sec:Discussion}

\subsection{From topology to model-independent predictions}
In this study, we investigated the deterministic dynamics of two interacting species under density-dependent interactions, with a particular focus on mutualistic interactions. 
We developed a general theoretical framework, formalized by Theorem \ref{thm:equilibrium}. These results are independent from the specific type of ecological interaction considered, covering for instance predator--prey, competitive interactions as well as a wide range of mutualistic interactions, as highlighted by the representative models summarized in Appendix~\ref{app:mutualism_models}.\\

The main mathematical contribution of Theorem \ref{thm:equilibrium} is not only to assign an index to each equilibrium point, but to show that these indices alternate along the isoclines under assumptions \ref{hyp:one} to \ref{hyp:four}.
This alternation, imposed by the crossing structure of the isoclines and by the sign changes of the vector field across them, provides a global constraint on the possible phase portraits of the system. 
In particular, attractive or repulsive coexistence equilibria of index $+1$ cannot accumulate along an isocline without being separated by saddle-type equilibria of index $-1$. 
Consequently, the qualitative organization of the phase portrait is partly determined by the geometry of the isoclines alone, independently of the precise functional form of the demographic rates.\\

Biologically, this index alternation can be interpreted as an alternation between potential coexistence states and ecological thresholds, \textit{i.e.} critical density configurations separating different long-term outcomes in the phase portrait.
In particular, a saddle equilibrium is associated with stable and unstable directions; its stable manifold may act as a separatrix between basins of attraction. 
Crossing such a threshold may determine whether the system converges to coexistence, to extinction of one species, or to another long-term state. 
Thus, when several coexistence equilibria are present along an isocline, Theorem~\ref{thm:equilibrium} does not only describe their local nature: it imposes that threshold-type equilibria must occur between consecutive non-saddle equilibria. 
This is a topological constraint on the ordering of equilibria, not an assumption on the particular functional form of the model.\\

Our index-theoretic approach therefore complements classical stability analyses: instead of studying each model separately, it identifies robust topological restrictions shared by all models satisfying the same qualitative assumptions.
Combined with additional assumptions on the signs of the relevant partial derivatives, these restrictions translate into predictions on stability and long-term outcomes, including coexistence, single-species
persistence, and the possible emergence of more complex dynamics such as periodic solutions.
Importantly, these conclusions do not rely on explicit functional forms of the density-dependent interaction functions, but only on the geometry of the isoclines and on the signs of the derivatives across different regions of the phase space. 
The framework therefore provides a unified methodology for studying how density-dependent costs and benefits can reshape both community equilibria and the ecological relationships between interacting species.

\subsection{Application to mutualism.}
To demonstrate the utility of our framework, we applied it to the analysis of mutualistic interactions. 
Most mutualistic models reviewed in \cite{hale2021ecological} satisfy assumptions \ref{hyp:one} to \ref{hyp:four}, indicating that our results are relevant to a large fraction of the theoretical literature on mutualism. 
Nevertheless, some models do not satisfy these assumptions, for instance by allowing multiple isoclines to intersect at a single equilibrium or by incorporating emigration dynamics \citep{thompson2006dynamics, holland2010}. 
Nevertheless, our framework extends the classical definition of mutualism by allowing density-dependent transitions between mutualistic and antagonistic interactions. 
Thus, it covers a wide spectrum of ecologically realistic situations and captures interaction dynamics that fall outside the scope of many traditional mutualistic models.
Our extension from the classical definition of mutualism typically includes realistic ecological interactions where, for each species, there is at least one region in the phase portrait where it benefits from the presence of the other species. 
This allows, for instance, the description of asymmetrical mutualistic interactions where one species benefits from the interaction over a wider range of densities than the other species. 
This relaxed definition enables the exploration of more diverse ecological dynamics and highlights biologically realistic features of species interactions that can crucially depend on the relative densities of interacting species, as observed for instance in density-dependent mutualistic effects in ant--aphid interactions \citep{breton1992density}. 
Even in textbook examples of mutualistic interactions, such as pollination of flowering plants by insects, high densities of invasive bees have been documented to decrease plant reproductive success because of the damage caused by over-visitation of flowers \citep{aizen2014mutualism}. 
As shown in some of the mutualistic models we studied, small changes in parameter values may cause an equilibrium point initially located in a strictly mutualistic region to shift into a parasitic one. 
By allowing such transitions between mutualism and parasitism, our extended framework sheds new light on the qualitative structure of species interactions and permits the emergence of dynamical phenomena not accessible under strictly mutualistic assumptions, including the theoretical possibility of limit cycles.

\subsection{Cyclic behaviours}
Numerous empirical studies report cyclic population dynamics in nature, which are overwhelmingly interpreted as the outcome of predator–prey or parasitic interactions \citep[p. 383-384]{turchin2003}.
Our results support this intuition, but also suggest that such cyclical dynamics might be observed throughout a larger array of density-dependent ecological interactions. 
Under strict mutualism, when the relevant signs of the partial derivatives remain fixed in a neighborhood of the equilibrium, our results exclude the local generation of limit cycles around that equilibrium.
By contrast, the cyclic configurations of Section~\ref{subsec:Extended_Mutualism} arise when the interaction changes along the trajectory.
These oscillations share some similarities with classical predator--prey cycles: part of the trajectory may pass through a region where one species benefits from the other while negatively affecting it in return.
However, they are not equivalent to predator--prey cycles with a fixed interaction sign. 
Here we found cyclic configurations where the limit cycle also passes through a region of strict mutualism, even though the repulsive equilibrium around which it forms is not itself located in that region. 
The cycle should therefore be interpreted as the outcome of a density-dependent change in interaction type along the trajectory: the system alternates between a mutualistic phase and an antagonistic, predator--prey-like phase.\\
To our knowledge, three contexts can give rise to periodic oscillations in a mutualism model: the two general contexts described in Section \ref{subsec:Extended_Mutualism} and the model analyzed by \cite{song2025}. 
Note that the model of \cite{song2025} is a particular case of System \ref{eq:init} but it does not fall within the general configurations of Section \ref{subsec:Extended_Mutualism}, suggesting that periodic oscillations can occur in other situations.\\

It is striking, however, that  the cyclic configurations described here and the model of \cite{song2025} share a common feature: a change of sign in one intraspecific density-dependent term.
In Section \ref{subsec:Extended_Mutualism}, this sign change takes the form of an Allee-effect feedback, with positive intraspecific density dependence at low density followed by negative density dependence at higher density.
Nevertheless, the exclusion criterion given by Proposition \ref{thm:cycle_lim_2} is broader than the presence of an Allee effect.
Proposition~\ref{thm:cycle_lim_2} does not prove that an individual sign change of $\frac{\partial f}{\partial x}$ or $\frac{\partial g}{\partial y}$ is necessary in all cyclic systems.
Rather, it excludes limit cycles when the two intraspecific terms keep the same constant sign throughout the strictly positive quadrant.
Thus, non-monotone self-regulation appears as a recurrent mechanism in the cyclic examples considered here, while Allee-effect feedbacks provide an ecologically important particular case of this mechanism.
Overall, our results propose an alternative and complementary interpretative framework for explaining some of the cyclic patterns observed in natural systems.\\

The absence of limit cycles in strict mutualism should therefore be understood as a structural limitation of two-dimensional models with fixed strictly positive cross effects.
Such models do not contain, by themselves, the negative feedback loop that is usually responsible for sustained oscillations.
This suggests that, if sustained cycles were to be observed in systems that are empirically classified as mutualistic, additional ecological structure may be required to generate them. 
Possible mechanisms include explicit resource dynamics and benefit--cost feedbacks \citep{holland2002, nathaniel2009consumer, holland2010}, demographic structure, spatial structure or dispersal \citep{thompson2006dynamics, aliyu2021interplay}, environmental feedbacks, slow evolutionary changes in interaction parameters, or the presence of more than two interacting species \citep{bascompte2003nested, martignoni2020parasitism}. 
In this sense, a two-dimensional density-dependent model may also be interpreted as a reduced effective description of a more complex ecological system.
\\

To our knowledge, empirical evidence documenting sustained population cycles driven by mutualistic interactions has not been reported, but long-term surveys of mutualistic populations with precise density estimates are actually scarce.
Given the current lack of precise empirical data on the temporal dynamics of mutualistic populations, our results should primarily be interpreted as theoretical possibilities, which might nonetheless become relevant under particular ecological or evolutionary scenarios yet to be documented. 
Our framework thus highlights potential behaviors that could guide future empirical investigations aimed at identifying conditions under which mutualistic interactions may fluctuate rather than converge to equilibrium. 
\section{Proofs}
\label{sec:Proofs}

In this section, we provide the proofs of the different results stated in Sections \ref{sec:General_Behavior} and \ref{sec:Mutualistic}. 
We begin with the proof of Theorem \ref{thm:equilibrium} and Proposition \ref{prop:equilibrium_axes}, which describe the alternation of equilibrium point indices along the isoclines.

\begin{proof}[Proof of Theorem \ref{thm:equilibrium}]
Let us consider the dynamical system \eqref{eq:init}, where $f$ and $g$ satisfy hypotheses \ref{hyp:one} to \ref{hyp:four}.  

From Hypothesis \ref{hyp:one}, the strictly positive quadrant contains exactly two continuous curves $\Gamma_f$ and $\Gamma_g$, which may intersect. 
These curves partition the strictly positive quadrant into distinct regions.
According to \ref{hyp:two}, inside each zone, the signs of $f$ and $g$, therefore those of $\dot{x}$ and $\dot{y}$ (for $x>0$ and $y>0$), remain constant. 
Thus, each region in  $\R_+ \times \R_+$ can be associated with a unique pair of signs $\left(\dot{x}, \dot{y}\right)$.
There are four possible pairs of signs: $\left(+,+\right)$, $\left(-,-\right)$, $\left(+,-\right)$ and $\left(-,+\right)$ for $\left(\dot{x}>0, \dot{y}>0\right)$, $\left(\dot{x}<0, \dot{y}<0\right)$, $\left(\dot{x}>0, \dot{y}<0\right)$ and $\left(\dot{x}<0, \dot{y}>0\right)$.

By \ref{hyp:three}, equilibrium points correspond to the intersections of $\Gamma_f$ and $\Gamma_g$. 
Since there are no more than two curves intersecting at any given point, the plane near an equilibrium point is divided into four regions.
To move from one zone to another, crossing a $\Gamma_f$ curve causes a change in the sign of $\dot{x}$, crossing a $\Gamma_g$ curve causes a change in the sign of $\dot{y}$.

The assumption that the sign of either $\dot{x}$ or $\dot{y}$ changes on either side of an isocline (\ref{hyp:two}) ensures that at least one component of the pair $\left(\dot{x}, \dot{y}\right)$ changes sign.
By Assumption \ref{hyp:four}, since the isoclines do not coincide, both components of $\left(\dot{x}, \dot{y}\right)$ cannot change sign at the same time.
From these properties, near an equilibrium point, there are exactly two possible clockwise sequences of sign changes:  

\begin{itemize}
    \item[(i)] $\left(+,+\right)$, $\left(+,-\right)$, $\left(-,-\right)$ and $\left(-,+\right)$,
    \item[(ii)] $\left(+,+\right)$, $\left(-,+\right)$, $\left(-,-\right)$ and $\left(+,-\right)$.
\end{itemize} 

Each pair corresponds to a specific direction of the vector field.
We can thus associate four phase portraits with each sign alternation, representing the rotation of sign pairs near the equilibrium point in Figure \ref{fig:PhasePortrait_Index}. 

Using this information, we can determine the index of the vector field around the equilibrium  point.
Given the form of the vector field in Figure \ref{fig:PhasePortrait_Index-1}, an equilibrium point with the sign alternation pattern (ii) corresponds to an index of $-1$, indicating a saddle point.
Similarly, according to the form of the vector field in Figure \ref{fig:PhasePortrait_Index+1},  an equilibrium point with the sign alternation pattern (i) corresponds to an index of $+1$, which can represent an attractive point, a repulsive point, or even a center.

Finally, \ref{hyp:one} ensures that only two isoclines, $\Gamma_f$ and $\Gamma_g$, exist in the strictly positive quadrant, making equilibrium points their intersections. 
\ref{hyp:four} requires that the curves alternate at each intersection due to their crossing.
Consequently, equilibrium points with an index of $+1$ are always followed by points with an index of $-1$ along the isoclines, and \textit{vice versa}.
\end{proof}

\begin{proof}[Proof of Proposition \ref{prop:equilibrium_axes}]
We restrict the analysis to the case where the equilibrium point lies on the $x$-axis \textit{i.e.} is generated by an intersection of the isocline $\Gamma_f$ with the axis $y = 0$.
The case of an equilibrium point located on the $y$-axis, arising from an intersection of $\Gamma_g$ with the axis $x = 0$, follows by symmetry.

Let $(x_1,0)$ denote the equilibrium point closest to the origin $(0,0)$.
Under our assumptions, there may exist at most one additional equilibrium point on the $x$-axis, denoted by $(x_2,0)$, with $x_1 < x_2$, which may or may not exist.
We keep the same notations and conventions as in the proof of Theorem \ref{thm:equilibrium}.

In a neighborhood of the equilibrium point $(x_1,0)$, four local configurations of the vector field may occur, depending on the associated sign of the vector field across the adjacent regions.
The equilibrium point $(x_1,0)$ may have index $+1$, corresponding to an attractive or repulsive equilibrium.
In this case, the alternation of signs across the four regions is either $(+,-),(-,-)$ or $(-,+),(+,+)$ illustrated in Figure \ref{fig:demo_prop}(\subref{fig:axe_Index+1}).
Alternatively, $(x_1,0)$ may have index $-1$, in which case it is a saddle point, and the alternation of signs is either $(+,+),(-,+)$ or $(-,-),(+,-)$ illustrated in Figure \ref{fig:demo_prop}(\subref{fig:axe_Index-1}).
As in the interior case, the ordering of signs is read clockwise.

\begin{figure}[h]
    \centering
    \begin{subfigure}[b]{0.45\textwidth}
        \centering
        \includegraphics[width=\textwidth]{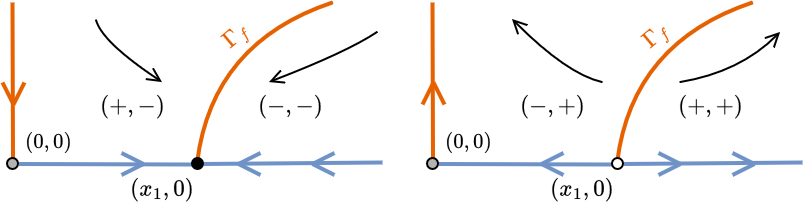}
        \caption{Local possible structure of the phase portrait around an index $+1$ equilibrium on the $x$-axis.}
        \label{fig:axe_Index+1}
    \end{subfigure}
    \hfill
    \begin{subfigure}[b]{0.45\textwidth}
        \centering
        \includegraphics[width=\textwidth]{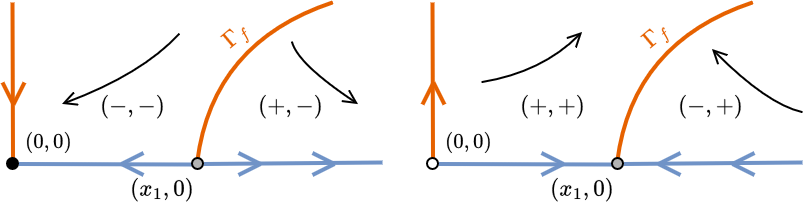}
        \caption{Local possible structure of the phase portrait around an index $-1$ equilibrium on the $x$-axis.}
        \label{fig:axe_Index-1}
    \end{subfigure}
    \caption{\textit{Four possible local configurations of the phase portrait near an equilibrium point located on the $x$-axis.} Two possible configurations are shown for an equilibrium point with index $+1$ (\ref{fig:axe_Index+1}) and index $-1$ (\ref{fig:axe_Index-1}). The alternation of signs of the vector field across adjacent regions is read clockwise. These configurations correspond to equilibrium points generated by intersections of the isocline $\Gamma_f$ with the axis $y = 0$.}
    \label{fig:demo_prop}
\end{figure}

Thus, when following the isocline $\Gamma_f$ and encountering an equilibrium point located in the strictly positive quadrant, the first intersection of $\Gamma_f$ with $\Gamma_g$ reached when moving along $\Gamma_f$ starting from $(x_1,0)$, the alternation of signs is reversed, while still being read clockwise.
More precisely, if the equilibrium point $(x_1,0)$ on the $x$-axis is characterized by the alternation $(+,-),(-,-)$ in Figure \ref{fig:axe_Index+1}, then the first equilibrium point encountered in the strictly positive quadrant exhibits the alternation $(+,+)$, $(-,+)$, $(-,-)$, and $(+,-)$, corresponding to case $(ii)$ in the proof of Theorem \ref{thm:equilibrium}, which yields an equilibrium point with index $-1$.
The remaining cases follow the same principle.
The remaining alternations in the strictly positive quadrant then follow from Theorem \ref{thm:equilibrium}.

It remains to determine the nature of the equilibrium point $(x_2,0)$, when it exists.
At this stage, we have shown that the index alternation persists when the equilibrium point $(x_1,0)$ on the $x$-axis is taken into account.
It remains to determine whether this alternation is also satisfied by a second equilibrium point on the axis, if such a point exists.
This depends on whether the isocline $\Gamma_g$ intersects the $x$-axis, and on the number of such intersections.
Indeed, even if intersections of $\Gamma_g$ with the $x$-axis do not generate additional equilibrium points, they modify the sign of the vector field in the vertical direction (for example changing a $(+,+)$ into a $(+,-)$).

If $\Gamma_g$ does not intersect the $x$-axis between $(x_1,0)$ and $(x_2,0)$, or if $\Gamma_g$ intersects it twice on this interval (resulting in a double sign change), then the alternation  at $(x_2,0)$ is the reverse of that at $(x_1,0)$.
For instance, if the alternation at $(x_1,0)$ is $(+,+)$ and $(-,+)$ in Figure \ref{fig:axe_Index-1}, then at $(x_2,0)$ it becomes $(-,+)$ and $(+,+)$, which transforms an equilibrium point with index $-1$ into one with index $+1$.
It remains to verify whether this behavior is consistent with the previously established index alternation.
In both situations (zero or two intersections of $\Gamma_g$ with $x$-axis), $\Gamma_g$ may generate equilibrium points in the phase portrait through its intersections with $\Gamma_f$.
Since the isoclines are assumed to be continuous and by Theorem \ref{thm:equilibrium}, the number of equilibrium points created by intersections of $\Gamma_f$ and $\Gamma_g$ before $\Gamma_f$ reaches $(x_2,0)$ is necessarily even.
Consequently, the last equilibrium point encountered in the strictly positive quadrant, corresponding to the final intersection of $\Gamma_f$ and $\Gamma_g$, has the same index as $(x_1,0)$.
As shown above, the equilibrium point $(x_2,0)$ has an index opposite to that of $(x_1,0)$.
Therefore, the alternation of equilibrium point indices is preserved when $\Gamma_g$ has zero or two intersections with the $x$-axis between $(x_1,0)$ and $(x_2,0)$.

It remains to consider the case where the isocline $\Gamma_g$ intersects the $x$-axis exactly once between $(x_1,0)$ and $(x_2,0)$.
By the same reasoning as above, for an equilibrium point $(x_1,0)$ characterized by a given alternation of sign pairs, the intersection of $\Gamma_g$ with the $x$-axis induces a change in the sign of the vector field in the $y$-direction along the axis.
Then, when reaching $(x_2,0)$, the intersection with $\Gamma_f$ induces a change in the sign of the vector field in the $x$-direction.
For instance, if the alternation around $(x_1,0)$ is given by $(+,+)$ and $(-,+)$, then after crossing $\Gamma_g$ the alternation becomes $(-,-)$, and after crossing $\Gamma_f$ it becomes $(+,-)$ at $(x_2,0)$.
As a result, the equilibrium point $(x_2,0)$ has the same index as $(x_1,0)$.
Moreover, since $\Gamma_g$ is continuous, it intersects the isocline $\Gamma_f$ at least once between $(x_1,0)$ and $(x_2,0)$, and in fact an odd number of times.
Consequently, the last equilibrium point generated by an intersection of $\Gamma_f$ and $\Gamma_g$ has an index opposite to that of $(x_1,0)$, and hence opposite to that of $(x_2,0)$.
Therefore, the alternation of equilibrium point indices is preserved when following $\Gamma_f$ from $(x_1,0)$ to $(x_2,0)$.

\end{proof}

We now provide the proof of Corollary \ref{cor:attrac}, which is a consequence of Theorem \ref{thm:equilibrium}. 
The proof of Corollary \ref{cor:repul} is very similar and will be omitted.

\begin{proof}[Proof of Corollary \ref{cor:attrac}]
Let $\left(x, y\right)$ be an equilibrium point in the strictly positive quadrant.
We can calculate the Jacobian of the system \eqref{eq:init} at the point $\left(x, y\right)$:

\begin{equation}
J\left(x,y\right) = 
\begin{pmatrix}
x \frac{\partial f\left(x,y\right)}{\partial x} & x \frac{\partial f\left(x,y\right)}{\partial y} \\
y \frac{\partial g\left(x,y\right)}{\partial x} & y \frac{\partial g\left(x,y\right)}{\partial y}
\end{pmatrix}.
\end{equation}

And the trace of $J\left(x,y\right)$ is:

\begin{equation}
\text{Tr}\left(J\left(x,y\right)\right) = x \frac{\partial f\left(x,y\right)}{\partial x}+ y \frac{\partial g\left(x,y\right)}{\partial y}<0.
\end{equation}

It implies that the point $\left(x,y\right)$ is an attractive equilibrium if $\text{det}\left(J\left(x,y\right)\right)>0$, else, if $\text{det}\left(J\left(x,y\right)\right)<0$, the point $\left(x,y\right)$ is a saddle point.
\end{proof}

We now prove Corollary \ref{cor:attrac_axes}. 
Notice that the proof of Corollary \ref{cor:repul_axes} will be omitted since it follows exactly the same arguments.

\begin{proof}
We suppose that $x_1>0$ and $(x_1,0)$ is an equilibrium point on the $x$-axis, that is, $f(x_1,0)=0$.
We prove the result for equilibria on the $x$-axis, \textit{i.e.} for equilibrium points given by the intersection of $\Gamma_f$ with the line $y=0$.
The proof for equilibria on the $y$-axis, corresponding to the intersection of $\Gamma_g$ with the line $x=0$ is identical.
We recall that $f\in \cun$ and that we assume $\frac{\partial f}{\partial x}(x,y)<0$ for all $(x,y)$.

Since $\frac{\partial f}{\partial x} <0$ on the non-negative quadrant, the function $x\mapsto f(x,0)$ is strictly decreasing. 
Hence it can have at most one zero on the positive $x$-axis. 
Therefore, if an equilibrium exists on this axis, it
is unique.

Let $x_1$ be the unique equilibrium on the $x$-axis.
At such a boundary equilibrium $(x_1,0)$, one eigenvalue is $x_1\frac{\partial f}{\partial x}(x_1,0)<0$. If the equilibrium has index $+1$, the determinant is positive, hence the second eigenvalue is also negative. 
The equilibrium is therefore attractive.




\end{proof}

We now prove why stable limit cycles cannot be observed under assumptions of strict mutualism.

\begin{proof}[Proof of Theorem \ref{thm:no_cycle_lim_1}]
$\left(x^*, y^* \right)$ is a repulsive equilibrium point with the assumptions $\frac{\partial f}{\partial y} (x^*, y^*) > 0$ and $\frac{\partial g}{\partial x} (x^*, y^*) > 0$. 
According to Table \ref{tab:eq_iso} (see Appendix \ref{appendix:IsoTable}), this case only arises when, in addition, $\frac{\partial f}{\partial x} (x^*, y^*) > 0$ and $\frac{\partial g}{\partial y} (x^*, y^*) > 0$. 
Since the functions $f$ and $g$ are $\cun$, there exists a neighborhood $R_1$ around the point $(x^*, y^*)$ such that $\frac{\partial f}{\partial y} > 0$, $\frac{\partial g}{\partial x} > 0$, $\frac{\partial f}{\partial x} > 0$ and $\frac{\partial g}{\partial y}> 0$.

By the implicit function theorem, we conclude  that, in the neighborhood of the point $(x^*, y^*)$, both isoclines $\Gamma_f$ and $\Gamma_g$ are decreasing.
To correspond to a repulsive equilibrium point (with index $+1$), the sign alternation around this point must match that of Figure \ref{fig:PhasePortrait_Index+1}.
There are only two possible phase portrait configurations in the neighborhood see Figure \ref{fig:No_Cycle_Limit_1}. 

\begin{figure}[h]
    \centering
    \begin{subfigure}[b]{0.43\textwidth}
        \centering
        \includegraphics[width=0.8\textwidth]{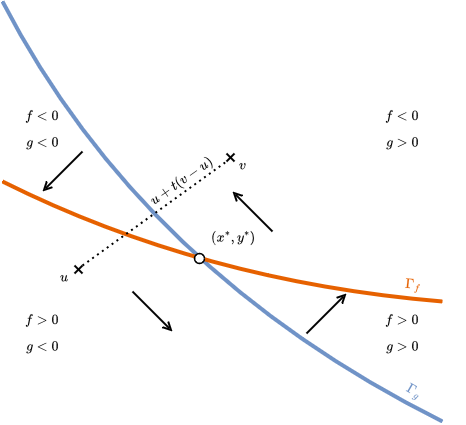}
        \caption{First possible alternation}
        \label{fig:No_Cycle_Limit_1_1}
    \end{subfigure}
    \hfill
    \begin{subfigure}[b]{0.43\textwidth}
        \centering
        \includegraphics[width=0.8\textwidth]{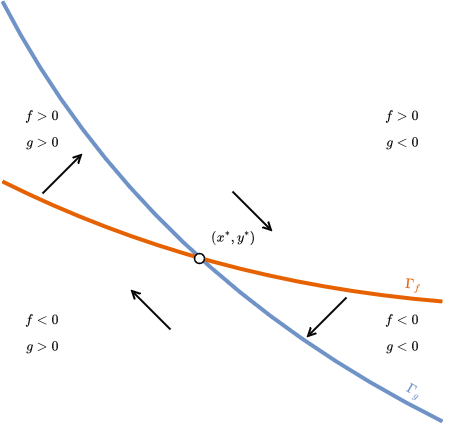}
        \caption{Second possible alternation}
        \label{fig:No_Cycle_Limit_1_2}
    \end{subfigure}
    \caption{\textit{Illustration of Proof of Theorem \ref{thm:no_cycle_lim_1}.} Two possible neighborhoods of the repulsive equilibrium point $(x^*, y^*)$ that could lead to the formation of a limit cycle around it.}
    \label{fig:No_Cycle_Limit_1}
\end{figure}

Let us consider the case of the Figure \ref{fig:No_Cycle_Limit_1_1}. 
The other case of the Figure \ref{fig:No_Cycle_Limit_1_2} follows similarly with $g$.  
We know the sign of $f$ and $g$ on either side of $\Gamma_f$ and $\Gamma_g$.  
Let $u$ and $v$ be points as shown in Figure \ref{fig:No_Cycle_Limit_1_1}, with respective coordinates 
$\begin{pmatrix} x_1 \\ y_1 \end{pmatrix}$ and $ \begin{pmatrix} x_2 \\ y_2 \end{pmatrix}$,
where $x_1 < x_2$ and $y_1 < y_2$, and define the function $h$ as follows:

\begin{align}
\label{eq:h}
h : [0, 1] &\to \mathbb{R} \\
t &\mapsto f(u + t(v - u))
\end{align}

Then, we have:  

\begin{align}
\label{eq:h_sign_0}
h(0) &= f(u) > 0, \\
\label{eq:h_sign_1}
h(1) &= f(v) < 0.
\end{align}

By construction, $h$ is continuously differentiable $\cun$, with derivative given by: 

\begin{equation*}
h'(t) = (x_2 - x_1) \frac{\partial f}{\partial x} (u + t(v - u)) + (y_2 - y_1) \frac{\partial f}{\partial y} (u + t(v - u)).
\end{equation*}

However, on the domain of definition of $h$, we have $\frac{\partial f}{\partial x} > 0$ and $\frac{\partial f}{\partial y} > 0$.  
Therefore, $h'(t) > 0$ for all $t \in \left[0, 1\right]$, which is impossible since $h$ is $\cun$ and by \eqref{eq:h_sign_0} and \eqref{eq:h_sign_1}.

\end{proof}

The proof of Theorem~\ref{thm:no_cycle_lim_2} follows the same approach as Theorem~\ref{thm:no_cycle_lim_1}, with appropriate changes in the choice of points and in the direction of the isoclines.

\begin{proof}[Proof of Theorem \ref{thm:no_cycle_lim_2}]

$\left(x^*, y^* \right)$ is an attractive equilibrium point with the assumptions $\frac{\partial f}{\partial y} (x^*, y^*) > 0$ and $\frac{\partial g}{\partial x} (x^*, y^*) > 0$. 
According to Table \ref{tab:eq_iso} (see Appendix \ref{appendix:IsoTable}), this case only arises when, in addition, $\frac{\partial f}{\partial x} (x^*, y^*) < 0$ and $\frac{\partial g}{\partial y} (x^*, y^*) < 0$. 
Since the functions $f$ and $g$ are $\cun$, there exists a neighborhood $R_2$ around the point $(x^*, y^*)$ such that $\frac{\partial f}{\partial y} > 0$, $\frac{\partial g}{\partial x} > 0$, $\frac{\partial f}{\partial x} < 0$ and $\frac{\partial g}{\partial y}< 0$.
Then both isoclines $\Gamma_f$ and $\Gamma_g$ are increasing.
To correspond to an attractive equilibrium point (with index $+1$), the sign alternation around this point must match that of Figure \ref{fig:PhasePortrait_Index+1}.
There are only two possible phase portrait configurations in the neighborhood see Figure \ref{fig:No_Cycle_Limit_2}. 

\begin{figure}[h]
    \centering
    \begin{subfigure}[b]{0.43\textwidth}
        \centering
        \includegraphics[width=0.8\textwidth]{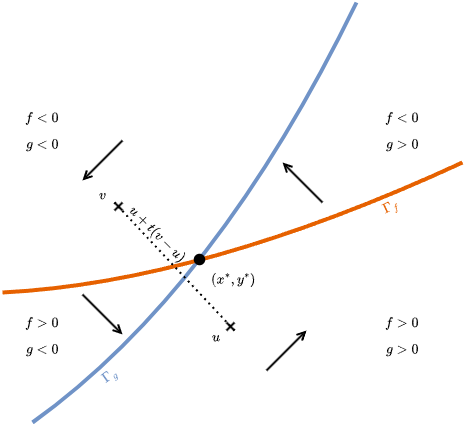}
        \caption{First possible alternation}
        \label{fig:No_Cycle_Limit_2_1}
    \end{subfigure}
    \hfill
    \begin{subfigure}[b]{0.43\textwidth}
        \centering
        \includegraphics[width=0.8\textwidth]{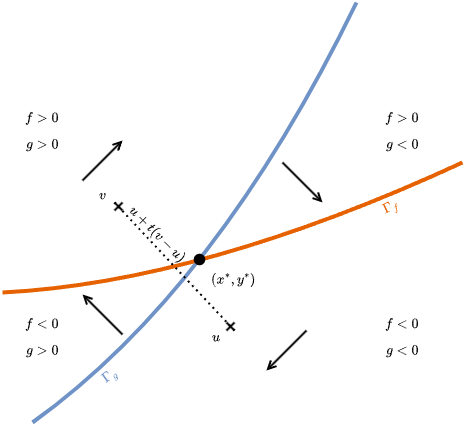}
        \caption{Second possible alternation}
        \label{fig:No_Cycle_Limit_2_2}
    \end{subfigure}
    \caption{\textit{Illustration of Proof of Theorem \ref{thm:no_cycle_lim_2}.} Two possible neighborhoods of the attractive equilibrium point $(x^*, y^*)$ that could lead to the formation of a repulsive limit cycle around it.}
    \label{fig:No_Cycle_Limit_2}
\end{figure}

Let us consider the case of the Figure \ref{fig:No_Cycle_Limit_2_1}. 
The other case of the Figure \ref{fig:No_Cycle_Limit_2_2} follows similarly with $g$.  
We know the sign of $f$ and $g$ on either side of $\Gamma_f$ and $\Gamma_g$.  
Let $u$ and $v$ be points as shown in Figure \ref{fig:No_Cycle_Limit_2_1}, with respective coordinates 
$\begin{pmatrix} x_1 \\ y_1 \end{pmatrix}$ and $ \begin{pmatrix} x_2 \\ y_2 \end{pmatrix}$,
where $x_1 > x_2$ and $y_1 < y_2$, and we define the function $h$ as in the previous proof, following \eqref{eq:h}, and obtain similar results as in \eqref{eq:h_sign_0} and \eqref{eq:h_sign_1}.  
Since $h$ is $\cun$ and has the same derivative as in the previous proof, we arrive at the same contradiction.

\end{proof}

We present the proof of Proposition \ref{thm:cycle_lim_1}, which provides conditions under which stable limit cycles may arise.

\begin{proof}[Proof of Proposition \ref{thm:cycle_lim_1}]

The proof is based on the Poincaré-Bendixson theorem \citep{ciesielski2012poincare}, used here through the fivefold way formulation introduced by \cite{coleman1983biological}.
Our strategy consists in constructing a positively invariant region of the phase plane and analyzing the possible $\omega$-limit sets of a trajectory entering this region.

For the first case, we construct a rectangular region composed of sides $A$, $B$, $C$, and $D$, as shown in Figure \ref{fig:cycle_lim_1}.
We will demonstrate that once a trajectory $\gamma$ enters this region, it cannot exit.

$A$ lies along a trajectory, defined by $x=0$, $y>0$, and $(0,0)$ is an equilibrium point. 
Thus, no trajectory can escape across $A$ as $t \to +\infty$, because to do so, it would have to intersect another trajectory or $(0,0)$, and trajectories cannot intersect.
$B$ consists of: two critical points, located at $(0,0)$ and $(x_2,0)$, the trajectory defined by $0 < x < x_2$, $y=0$, which moves away from $(0,0)$ towards the point $(x_2,0)$, and finally the trajectory defined by $x_2<x$, $y=0$, which moves towards the point $(x_2,0)$.
As before, no trajectory can exit across $B$.
Trajectories cross $C$ moving from right to left as $t$ increases, since the points along this arc are located above $\Gamma_f$. 
Given assumption 1.10 where $f<0$, this implies that $\dot{x} < 0$ along the arc.
The trajectory through the corner $e$ is horizontal and is moving to the left at that point, as $e$ lies on $\Gamma_g$, which means $\dot{y} = 0$. 
Along $D$, trajectories move downward and to the left, since the points on this arc lie above both isoclines. 
Therefore, no trajectories can escape across the boundary of the rectangular region, and trajectories enter the region through the points along the top and right edges. 

Our goal now is to demonstrate that the limit set $\omega(\gamma)$ of the trajectory $\gamma$ is a limit cycle.
According to the Poincaré-Bendixson theorem, stated in the form of the fivefold way (see \cite{coleman1983biological}, page 265), this will follow if we can show that $\omega(\gamma)$ contains no critical points.
The critical points of system \eqref{eq:init} are located at the origin $(0,0)$, at the point $(x_2, 0)$ on the $x$-axis, and at the point $\left(x^*,y^*\right)$.
$\left(x^*,y^*\right)$ cannot be in $\omega(\gamma)$ because $\left(x^*,y^*\right)$ is repulsive by assumption 1.9. 
Moreover we cannot have that $\omega(\gamma) = \{(0,0)\}$.
In fact, for the trajectory $\gamma$ to approach $(0,0)$, it must intersect the interior of the region where $f > 0$ and  $g < 0$ (as shown in Figure \ref{fig:cycle_lim_1}), and stay in this region as $t \to +\infty$.
However, if $\gamma$ is in the interior of this region, then $f > 0$ and $\dot{x}> 0$, which implies that we cannot have both $x(t)$, $y(t)$ remaining in the region for all large enough $t$ and $(x(t), y(t)) \rightarrow (0,0)$ as $t \rightarrow +\infty$. 
A similar argument shows that $\omega(\gamma)\neq \{(x_2,0)\}$.

Thus the last case is if $\omega(\gamma)$ contains any critical points at all, $\omega(\gamma)$ can be a cycle graph. 
$\left(x^*,y^*\right)$ cannot belong to a cycle graph since $\left(x^*,y^*\right)$ is not the $\omega$-limit set of any nonconstant trajectory.  
If $(x_2,0)$ belongs to a cycle graph, then so must the orbit along $B$ with $\omega$-limit set $\{(x_2,0)\}$ and $\alpha$-limit set $\{(0,0)\}$, since that is the only orbit in the rectangular region with $\omega$-limit set $\{(x_2,0)\}$.
But in a cycle graph, each limit set of each orbit is a critical point of the cycle graph. 
We must then have that $\{(0,0)\}$ in the cycle graph, then the unbounded orbit on $A$ also belongs to the cycle graph since it is the only orbit in the positive quadrant whose $\omega$-limit set is $\{(0,0)\}$. 
This is impossible since $\omega(\gamma)$
is bounded because $\gamma$ is bounded. 
Thus $\omega(\gamma)$ contains no critical points.
By the fivefold way (see \cite{coleman1983biological}, page 265), or equivalently by the Poincaré-Bendixson theorem in this formulation, we have that $\omega(\gamma)$ is a cycle.
\newline

We will follow the same approach for the second case.
We construct a rectangular region composed of sides $A$, $B$, $C$, and $D$, as shown in Figure \ref{fig:cycle_lim_2}.
We will demonstrate that once a trajectory $\gamma$ enters this region, it cannot exit.

$A$ consists of: two critical points, located at $(0,0)$ and $(0,y_2)$, the trajectory defined by $0 < y < y_2$, $x=0$, which moves away from $(0,0)$ towards the point $(0,y_2)$, and finally the trajectory defined by $y_2<y$, $x=0$, which moves towards the point $(0,y_2)$.
Thus, no trajectory can escape across $A$ as $t \to +\infty$, because to do so, it would have to intersect another trajectory or $(0,0)$.
$B$ lies along a trajectory, defined by $x>0$, $y=0$, and $(0,0)$ is an equilibrium point. 
As before, no trajectory can exit across $B$.
Trajectories cross $C$ moving from right to left as $t$ increases, since the points along this arc are located below $\Gamma_f$. 
Given assumption 2.9 where $f<0$, this implies that $\dot{x} < 0$ along the arc.
The trajectory through the corner $e$ is vertical and is moving downwards at that point, as $e$ lies on $\Gamma_f$. 
Along $D$, trajectories move downward and to the right, since the points on this arc lie above both isoclines. 
Therefore, no trajectories can escape across the boundary of the rectangular region, and trajectories enter the region through the points along the top and right edges. 

As in the previous proof, we will demonstrate that the limit set $\omega(\gamma)$ of the trajectory $\gamma$ is a limit cycle.
It will follow if we can show that $\omega(\gamma)$ contains no critical points.
The critical points of system \eqref{eq:init} are located at the origin $(0,0)$, at the point $(0, y_2)$ on the $y$-axis, and at the point $\left(x^*,y^*\right)$.
$\left(x^*,y^*\right)$ cannot be in $\omega(\gamma)$ because $\left(x^*,y^*\right)$ is repulsive by assumption 2.8. 
Moreover we cannot have that $\omega(\gamma) = \{(0,0)\}$.
For the trajectory $\gamma$ to approach $(0,0)$, it must intersect the interior of the region where $f < 0$ and  $g > 0$ (as shown in Figure \ref{fig:cycle_lim_2}), and stay in this region as $t \to +\infty$, which is impossible as before and a similar argument shows that $\omega(\gamma)\neq \{(0,y_2)\}$.

Then the last case is if $\omega(\gamma)$ contains any critical points at all, $\omega(\gamma)$ can be a cycle graph. 
$\left(x^*,y^*\right)$ cannot belong to a cycle graph as before.  
If $(0,y_2)$ belongs to a cycle graph, then so must the orbit along $A$ with $\omega$-limit set $\{(0,y_2)\}$ and $\alpha$-limit set $\{(0,0)\}$, since that is the only orbit in the rectangular region with $\omega$-limit set $\{(0,y_2)\}$.
But in a cycle graph, each limit set of each orbit is a critical point of the cycle graph. 
We must then have that $\{(0,0)\}$ in the cycle graph, then the unbounded orbit on $B$ also belongs to the cycle graph since it is the only orbit in the positive quadrant whose $\omega$-limit set is $\{(0,0)\}$. 
This is impossible since $\omega(\gamma)$ is bounded because $\gamma$ is bounded. 
Thus $\omega(\gamma)$ contains no critical points.
By the fivefold way, we have that $\omega(\gamma)$ is a cycle.
\end{proof}

We conclude with the proof of Proposition \ref{thm:cycle_lim_2}, which provides a necessary condition for the existence of a limit cycle.

\begin{proof}
We consider system \eqref{eq:init} on the strictly positive quadrant:

\begin{equation*}
\left\{
    \begin{array}{ll}
        \dot{x} = x f(x,y), \\
        \dot{y} = y g(x,y).
    \end{array}
\right.
\end{equation*}

We assume that $\frac{\partial f}{\partial x}$ and $\frac{\partial g}{\partial y}$ are both non-positive everywhere, or both non-negative everywhere, and are not both identically zero in the strictly positive quadrant.

We consider the Dulac function
\[
B(x,y) = \frac{1}{xy},
\]
which is $C^1$ on the strictly positive quadrant.

Let $F(x,y) = (x f(x,y),\, y g(x,y))$ denote the vector field. 
Then
\[
B F = \left( \frac{f(x,y)}{y},\, \frac{g(x,y)}{x} \right).
\]
Its divergence is
\[
\nabla \cdot (B F)
= \frac{\partial}{\partial x} \left( \frac{f(x,y)}{y} \right)
+ \frac{\partial}{\partial y} \left( \frac{g(x,y)}{x} \right)
= \frac{1}{y} \frac{\partial f}{\partial x}
+ \frac{1}{x} \frac{\partial g}{\partial y}.
\]

Since $x>0$ and $y>0$, the sign of 
$\nabla \cdot (B F)$ is entirely determined by the signs of $\frac{\partial f}{\partial x}$ and $\frac{\partial g}{\partial y}$. 
If both $\frac{\partial f}{\partial x}$ and $\frac{\partial g}{\partial y}$ are non-positive everywhere, or both non-negative everywhere, and are not both identically zero, then $\nabla \cdot (B F)$ does not change sign and is not identically zero.

By the Bendixson--Dulac criterion (see \cite{perko2013differential}), a $C^1$ planar vector field defined on the strictly positive quadrant cannot admit a periodic orbit entirely contained in that domain if there exists a $C^1$ Dulac function whose weighted divergence does not change sign.

Therefore, the system admits no limit cycle in the strictly positive quadrant.
\end{proof}

\section*{Acknowledgments}

We sincerely thank Frédéric Faure for his valuable discussions and teachings.
This work has benefited from the support of the MMB Chair (Mathematics for Life and Environmental Sciences) VEOLIA-Ecole Polytechnique-MNHN-F.X.
The European Union (ERC Consolidator) grant OUTOFTHEBLUE (project number 101088089) obtained by VL supported this work. 
Views and opinions expressed are however those of the authors only and do not necessarily reflect those of the European Union or the European Research Council. 
Neither the European Union nor the granting authority can be held responsible for them.









\appendix

\section{Mutualistic models from the literature}
\label{app:mutualism_models}

The models summarized in this appendix illustrate how several classical mutualistic systems can be written in the general form~\eqref{eq:init}. 
They also show that mutualistic interactions may differ in their degree of dependence, specialization, symmetry, and density dependence.

Some mutualistic interactions are obligate: the persistence of each species then depends on the mutualistic relationship and hence on the presence of the interacting species, as in the specialized fig--wasp interaction \citep{bronstein2001}. 
In contrast, in facultative mutualism, each species may benefit from the interaction while still being able to persist in the absence of the other, as in generalist pollination systems where plants and pollinators can often find alternative partners \citep{sargent2008plant}. 
As a consequence, the functions describing the mutual influence between the two populations may take different forms depending on the level of dependence, specialization and symmetry between the interacting species.

For readability, we group the models into three classes according to their underlying ecological mechanisms: linear-benefit models, saturating models, and non-monotonic or consumer--resource models. 
Each model is expressed through the corresponding per-capita growth functions $f$ and $g$ appearing in system~\eqref{eq:init}.

\subsection{Linear–Benefit Mutualism Models}
\label{subsec:linear}

Table \ref{tab:linear} compiles mutualism models where benefit increases linearly with partner density and there is no saturation effect.

\begin{table}[h!]
\centering
\begin{tabular}{
    >{\centering\arraybackslash}m{3cm}
    >{\centering\arraybackslash}m{5cm}
    >{\centering\arraybackslash}m{4.1cm}
    >{\centering\arraybackslash}m{2.2cm}
}
\hline
\textbf{Model} & \textbf{Per-capita growth functions} & \textbf{Interpretation} & \textbf{Reference} \\
\hline
\\
\textbf{Linear benefit with logistic limitation} &
\(
\begin{aligned}
f(x,y) &= r_1\left(\dfrac{K_1+\alpha_{12}y - x}{K_1}\right)\\
g(x,y) &= r_2\left(\dfrac{K_2+\alpha_{21}x - y}{K_2}\right)
\end{aligned}
\) &
Linear increase of partner benefit; modification of intrinsic growth. &
\cite{gause1935behavior} \\
\hline
\textbf{Symbiont--host linear model} &
\(
\begin{aligned}
f(x,y) &= r_1\frac{K_1+\alpha_{12}y - x}{K_1+\alpha_{12}y}\\
g(x,y) &= r_2\frac{K_2+\alpha_{21}x - y}{K_2}
\end{aligned}
\) &
Asymmetric symbiont--host structure; linear modification of carrying capacity; obligate for \(K_1=0\). &
\cite{whittaker1970communities} \\
\hline
\textbf{Lotka--Volterra mutualism} &
\(
\begin{aligned}
f(x,y) &= r_1 + \alpha_{12}y - \alpha_{11}x\\
g(x,y) &= r_2 + \alpha_{21}x - \alpha_{22}y
\end{aligned}
\) &
Classical linear mutualistic enhancement of per-capita growth; obligate when \(r_1/\alpha_{11}\le 0\). &
\cite{vandermeer1978varieties} \\
\hline
\textbf{Multiplicative linear benefit} &
\begin{equation*}
\begin{aligned}
f(x,y) &= r_1\left(1-\dfrac{x}{K_1}\right)\left(1+\dfrac{\alpha_{12}y}{K_1}\right) \\
g(x,y) &= r_2\left(1-\dfrac{y}{K_2}\right)\left(1+\dfrac{\alpha_{21}x}{K_2}\right)
\end{aligned}
\end{equation*} &
Partner increases effective carrying capacity multiplicatively; facultative. &
\cite{addicott1981stability} \\
\hline
\end{tabular}
\caption{Linear–benefit mutualism models expressed through per-capita growth functions \(f\) and \(g\).}
\label{tab:linear}
\end{table}

\newpage
\subsection{Saturating Mutualism Models (Interspecific Saturation)}
\label{subsec:satur}
Table \ref{tab:saturing} compiles mutualism models where the benefits gained by a species are a saturating function of the density of the other species.
\begin{table}[h!]
\centering
\begin{tabular}{
    >{\centering\arraybackslash}m{3cm}
    >{\centering\arraybackslash}m{5cm}
    >{\centering\arraybackslash}m{4.1cm}
    >{\centering\arraybackslash}m{2.2cm}
}
\hline
\textbf{Model} & \textbf{Per-capita growth functions} & \textbf{Interpretation} & \textbf{Reference} \\
\hline
\textbf{Birth limitation reduced by partner} &
\(
\begin{aligned}
f(x,y) &= r_1 - \frac{b_1 x}{1 + m_1 y} - d_1 x \\
g(x,y) &= r_2 - \frac{b_2 y}{1 + m_2 x} - d_2 y
\end{aligned}
\) &
Partner reduces density-dependent birth limitation; saturating effect with partner density; facultative. &
\cite{wolin1984models} \\
\hline
\textbf{Saturating forager mutualism} &
\(
\begin{aligned}
f(x,y) &= r_1 - c_1 x + \frac{b_{12} a_{12} y}{1 + a_{12} h_{12} y}\\
g(x,y) &= r_2 - c_2 y + \frac{b_{21} a_{21} x}{1 + a_{21} h_{21} x}
\end{aligned}
\) &
Handling time limits mutualistic gains; stabilizes coexistence; fits pollination/foraging systems. &
\cite{wright1989simple} \\
\hline
\textbf{Exponential saturation of benefits} &
\begin{equation*}
\begin{aligned}
f(x,y) &= r_{10} + (r_{11}-r_{10}) \\
&\times \left(1-e^{-k_1 y}\right) - a_1 x \\
g(x,y) &= r_{20} + (r_{21}-r_{20})\\
&\times\left(1-e^{-k_2 x}\right) - a_2 y
\end{aligned}
\label{eq:Graves}
\end{equation*}
 &
Benefits saturate exponentially with partner density. &
\cite{graves2006bifurcation} \\
\hline
\\
\textbf{Plant--pollinator with shared saturation} &
\(
\begin{aligned}
f(x,y) &= \frac{\eta \alpha y}{1+\alpha x+\alpha \beta y} - b - c x \\
g(x,y) &= \frac{\mu \alpha x}{1+\alpha x+\alpha \beta y} - d
\end{aligned}
\) &
Shared saturation of benefits; obligate mutualism. &
\cite{fishman2010plant} \\
\hline
\end{tabular}
\caption{Mutualism models with interspecific saturation of benefits expressed through per-capita growth functions $f$ and $g$.}
\label{tab:saturing}
\end{table}

\newpage
\subsection{Non-Monotonic Mutualism Models and Consumer-Resource Approach}
\label{subsec:non-mono}
Table \ref{tab:non-mono} compiles mutualism models where the benefits gained by a species are a non-monotonic function of the density of the other species.

\begin{table}[h!]
\centering
\begin{tabular}{
    >{\centering\arraybackslash}m{3cm}
    >{\centering\arraybackslash}m{5cm}
    >{\centering\arraybackslash}m{4.1cm}
    >{\centering\arraybackslash}m{2.2cm}
}
\hline
\textbf{Model} & \textbf{Per-capita growth functions} & \textbf{Interpretation} & \textbf{Reference} \\
\hline
\textbf{Quadratic interaction} &
\(
\begin{aligned}
f(x,y) &= r_1\left(c_1 - x - a_1 (y - b_1)^2\right)\\
g(x,y) &= r_2\left(c_2 - y - a_2 (x - b_2)^2\right)
\end{aligned}
\label{eq:Zhang}
\) &
Benefit maximized at intermediate partner density; allows transitions between mutualism and competition. &
\cite{zhang2003mutualism} \\
\hline
\textbf{Plant--mycorrhizae with cost} &
\(
\begin{aligned}
f(x,y) &= r_1\frac{K_1+\alpha_{12}y-x}{K_1+\alpha_{12}y} - a y \\
g(x,y) &= r_2\frac{K_2+\alpha_{21}x-y}{K_2}
\end{aligned}
\label{eq:Neuhauser}
\) &
Interspecific asymmetric cost; facultative mutualism. &
\cite{neuhauser2004mutualism} \\
\hline
\textbf{Protection mutualism with saturating mortality reduction} &
\begin{equation*}
\begin{aligned}
f(x,y) &= \rho_1 b_1\left(1-\frac{x}{S_1}\right) \\
        &-\left(d_{1\min}+\frac{d_{1\max}-d_{1\min}}{1+c_1 y}\right) \\
g(x,y) &= \rho_2b_2\left(1-\frac{y}{S_2+x}\right)\\
        &-\left(d_{2\min}+\frac{d_{2\max}-d_{2\min}}{1+c_2 x}\right)
\end{aligned}
\end{equation*} &
Protection reduces mortality at low partner density; net effects may shift with density. &
\cite{thompson2006dynamics} \\
\hline
\\
\textbf{Consumer--resource mutualism} &
\(
\begin{aligned}
f(x,y) &= r_1 + c_1 \frac{a_{12} y}{h_2 + y} \\
&- q_1 \frac{\beta_{12} y}{e_1 + x} - s_1 x \\
g(x,y) &= r_2 + c_2 \frac{a_{21} x}{h_1+x} \\
        &- q_2 \frac{\beta_{21} x}{e_2+y} - s_2 y
\end{aligned}
\label{eq:Holland}
\) &
Bidirectional consumer–resource structure; obligate when $r_1=0$. &
\cite{holland2010} \\
\hline
\\
\textbf{Ant--fungal garden mutualism} &
\(
\begin{aligned}
f(x,y) &= r_f\frac{a y^2}{b+a y^2} - r_c y - d_1 x \\
g(x,y) &= r_a x - d_2 y
\end{aligned}
\) &
Consumer--resource mutualism with nonlinear resource production. &
\cite{kang2011mathematical} \\
\hline
\textbf{Plant--mycorrhizae with variable costs} &
\begin{equation*}
\begin{aligned}
f(x,y) &= r_p + q_{hp}\frac{\alpha y}{d+x}\\
        &- q_{cp}\beta y - \mu_p x \\
g(x,y) &= q_{cm}\beta x
        - q_{hm}\frac{\alpha x}{d+x}- \mu_m y
\end{aligned}
\end{equation*} &
Explicit separation of benefit and cost terms; obligate mutualism. &
\cite{martignoni2020parasitism} \\
\hline
\textbf{Plant--pollinator consumer--resource model} &
\begin{equation*}
\begin{aligned}
f(x,y) &= b_P\left(f+\phi\frac{axy}{1+a h x+axy}\right)\\
        &- s_P x - d_P \\
g(x,y) &= b_A + \epsilon\frac{a x}{1+a h x}
        - s_A y - d_A
\end{aligned}
\end{equation*} &
Mechanistic pollination model. &
\cite{hale2021ecological} \\
\hline
\end{tabular}
\caption{Non-monotonic mutualism and consumer-resource models with peak benefit at intermediate partner density expressed through per-capita growth functions $f$ and $g$.}
\label{tab:non-mono}
\end{table}
\newpage
\section{Index Theory and Equilibrium Classification}
\label{appendix:Index}

Index theory is an interesting tool in the qualitative study of dynamical systems. 
It provides a way to classify equilibrium points and understand the possible global behavior of trajectories.
The index of an equilibrium point is a topological invariant that remains unchanged under continuous deformations of the vector field.

\subsection{Definition of the Index}

Let $(x^*, y^*)$ be an equilibrium point of a planar dynamical system. 
The index of this equilibrium point is defined by analyzing the behavior of the vector field along a simple closed curve surrounding $(x^*, y^*)$. 
More formally, the index is given by the total winding number of the vector field as one traverses the closed curve counterclockwise.

Mathematically, a system given by:

\begin{equation}
\left\{
    \begin{array}{ll}
        \dot{x} &= F\left(x,y\right) \\
        \dot{y} &= G\left(x,y\right)
    \end{array}
\right.
\label{eq:index}
\end{equation}

defines a vector field for which the slope:

\begin{equation}
\frac{dy}{dx} = \frac{G(x, y)}{F(x, y)}
\end{equation}

forms an angle with the $x$-axis given by: 

\begin{equation}
\varphi(x, y) = \arctan \left( \frac{G(x, y)}{F(x, y)} \right).
\end{equation}

Given this angle and a simple closed curve $\gamma$ in $\mathbb{R}^2$, the index of $\gamma$ , denoted $\text{Ind}(\gamma)$, is defined as: 

\begin{equation}
\text{Ind}(\gamma) = \frac{1}{2\pi} \oint_{\gamma} d\varphi(x, y).
\end{equation}

\subsection{Indices of Different types of Equilibrium Points} 

The index of an equilibrium point depends on its stability and the local structure of trajectories.

\begin{itemize}
    \item Saddle Point (Index $-1$): A saddle point is an equilibrium where trajectories approach along one direction and diverge along another. 
    As one follows a simple closed curve around the equilibrium, the direction of the vector field makes one full turn in the negative orientation, leading to an index of $-1$.
    \item Attractive Point (Index $+1$): An attractor is an equilibrium where all nearby trajectories converge.
    The vector field makes a full $2\pi$ rotation around the equilibrium, giving an index of $+1$.
    \item Repulsive Point (Index $+1$): A repeller is an equilibrium where all nearby trajectories move away. As with attractors, the vector field completes a full rotation around the point, resulting in an index of $+1$.
\end{itemize}

\subsection{Implications of Index Theory}

A key consequence of index theory is the Poincaré-Hopf theorem, which states that the sum of the indices of all equilibrium points within a bounded region must equal the Euler characteristic of the domain. 
This imposes constraints on the possible configurations of equilibria and their stability.

For example, in a dynamical system where the entire phase plane is enclosed within a region, the total index must sum to $+1$. 
This means that a system with a single repulsive equilibrium point cannot have only saddle points without additional attractors or limit cycles to balance the total index.

\section{Trace-Determinant Analysis: Classical study of the point $(0,0)$}
\label{appendix:$(0,0)$}

The system \eqref{eq:init} always has two isoclines ($y = 0$ and $x = 0$), making $(0,0)$ an equilibrium point. 
To determine its nature, we examine the Jacobian matrix at $(0,0)$, assuming $f(0,0) \neq 0$ and $g(0,0) \neq 0$:

\[
J(0,0) = \begin{pmatrix}
f(0,0) & 0 \\
0 & g(0,0)
\end{pmatrix}.
\]

with trace $\text{Tr}(J) = f(0,0) + g(0,0)$ and determinant $\text{det}(J) = f(0,0) g(0,0)$. 
Three cases arise:
\begin{itemize}
    \item[(i)] Repulsive node ($f(0,0)>0$, $g(0,0) > 0$): Both species grow at low densities, leading to population increase. 
    This suggests facultative mutualism or independent persistence. In phase portrait, trajectories move away from extinction.
    \item[(ii)] Attractive node ($f(0,0)<0$, $g(0,0) < 0$): Both species decline at low densities, leading to extinction. 
    It can also indicate obligate mutualism, where each species is highly dependent on the other, and without sufficient partner density, both face extinction.
    This may indicate obligate mutualism or Allee effects. In phase portrait, trajectories converge to extinction.
    \item[(iii)] Saddle point ($f(0,0) g(0,0) < 0$): One species can persist alone, while the other requires interaction to survive. 
    This may indicate an asymmetric dependency (facultative-obligate mutualism where one species depends on the other for survival). 
    Some trajectories lead to extinction, while others promote coexistence or growth
\end{itemize}

These cases determine whether extinction is stable or unstable and how interactions shape population dynamics at low densities.

\section{Isocline behaviors and Point classifications}
\label{appendix:IsoTable}

We study system $\eqref{eq:init}$.
We assume that the functions $f$ and $g$ are $\cun$, so the sign of their partial derivatives remains strictly constant in a neighborhood of the equilibrium points.  
By the implicit function theorem, we can determine whether the isoclines are increasing or decreasing near these equilibrium points, in addition to identifying their nature using the trace-determinant diagram.
We classify the results based on the signs of the partial derivatives in Table \ref{tab:eq_iso} and specify when the system satisfies the assumptions of strict mutualism  as defined in Section \ref{subsec:Key_Assumptions_Def}.

\begin{sidewaystable}[h!]
\centering
    \begin{tabular}{|c|c|c|c|c|}
        \hline
         &  &  &  & \\
         & $\frac{\partial g}{\partial x}>0$ & $\frac{\partial g}{\partial x}>0$ & $\frac{\partial g}{\partial x}<0$ & $\frac{\partial g}{\partial x}<0$ \\
         &  &  &  & \\
         & $\frac{\partial g}{\partial y}>0$ & $\frac{\partial g}{\partial y}<0$ & $\frac{\partial g}{\partial y}>0$ & $\frac{\partial g}{\partial y}<0$ \\
         &  &  &  & \\
         \hline
         &  &  &  & \\
         $\frac{\partial f}{\partial x}>0$ & \textbf{Repulsive} or \textbf{Saddle Point} & \textbf{Saddle Point} & \textbf{Repulsive} & Indeterminate \\
         & $x$-isocline $\searrow$ & $x$-isocline $\searrow$ & $x$-isocline $\searrow$ & $x$-isocline $\searrow$ \\
         $\frac{\partial f}{\partial y}>0$ & $y$-isocline $\searrow$ & $y$-isocline $\nearrow$ & $y$-isocline $\nearrow$ & $y$-isocline $\searrow$ \\
         & \textit{Strict Mutualism} & \textit{Strict Mutualism} &  & \\
         \hline
         &  &  &  & \\
         $\frac{\partial f}{\partial x}>0$ & \textbf{Repulsive} & Indeterminate & \textbf{Repulsive} or \textbf{Saddle Point} & \textbf{Saddle Point} \\
         & $x$-isocline $\nearrow$ & $x$-isocline $\nearrow$ & $x$-isocline $\nearrow$ & $x$-isocline $\nearrow$ \\
         $\frac{\partial f}{\partial y}<0$ & $y$-isocline $\searrow$ & $y$-isocline $\nearrow$ & $y$-isocline $\nearrow$ & $y$-isocline $\searrow$ \\
         &  &  &  & \\
         \hline
         &  &  &  & \\
         $\frac{\partial f}{\partial x}<0$ & \textbf{Saddle Point} & \textbf{Attractive} or \textbf{Saddle Point} & Indeterminate & \textbf{Attractive} \\
         & $x$-isocline $\nearrow$ & $x$-isocline $\nearrow$ & $x$-isocline $\nearrow$ & $x$-isocline $\nearrow$ \\
         $\frac{\partial f}{\partial y}>0$ & $y$-isocline $\searrow$ & $y$-isocline $\nearrow$ & $y$-isocline $\nearrow$ & $y$-isocline $\searrow$ \\
         & \textit{Strict Mutualism} & \textit{Strict Mutualism} &  & \\
         \hline
         &  &  &  & \\
         $\frac{\partial f}{\partial x}<0$ & Indeterminate & \textbf{Attractive} & \textbf{Saddle Point} & \textbf{Attractive} or \textbf{Saddle Point} \\
         & $x$-isocline $\searrow$ & $x$-isocline $\searrow$ & $x$-isocline $\searrow$ & $x$-isocline $\searrow$ \\
         $\frac{\partial f}{\partial y}<0$ & $y$-isocline $\searrow$ & $y$-isocline $\nearrow$ & $y$-isocline $\nearrow$ & $y$-isocline $\searrow$ \\
         &  &  &  & \\
         \hline
    \end{tabular}
     \caption{Summary of isocline behaviors and point classifications.}
    \label{tab:eq_iso}
\end{sidewaystable}

\clearpage
\bibliographystyle{apalike}
\bibliography{references}

\end{document}